\documentclass{article}

\usepackage[margin=1.5cm]{geometry}

\usepackage{tikz}
\usetikzlibrary{arrows.meta, positioning, shapes.geometric}

\usepackage{amsmath,amssymb,amsfonts,mathtools,mathrsfs,bm}
\usepackage{amsthm} 
\usepackage{tabularx,booktabs,array,longtable,multirow,threeparttable}
\usepackage{graphicx,epsfig,subfigure}
\usepackage{color,colortbl} 
\usepackage{caption}
\usepackage{rotating}
\usepackage{diagbox}
\usepackage{pdflscape,lscape}
\usepackage{float}
\usepackage[dvipsnames]{xcolor}

\usepackage{cite,url,hyperref}
\hypersetup{
  colorlinks=true,
  linkcolor=lightblue,
  urlcolor=lightblue,
  citecolor=lightblue
}

\usepackage{algorithm,algpseudocode}

\usepackage{pifont}
\usepackage{tcolorbox}
\tcbuselibrary{skins}
\usepackage{multicol}

\definecolor{lightblue}{RGB}{70,130,180}   

\newcommand{\RMSection}[2]{\textbf{\hyperref[#2]{#1}}\par}
\newcommand{\RMSub}[2]{\hspace*{1.5em}{\color{gray}\footnotesize$\triangleright$}~\hyperref[#2]{#1}\par}

\theoremstyle{plain}
\newtheorem{theorem}{Theorem}
\newtheorem{proposition}[theorem]{Proposition} 

\theoremstyle{definition}

\theoremstyle{remark}

\title{Making Gaussian Kolmogorov--Arnold Networks Reliable and Accurate}

\author{
Amir Noorizadegan$^{1,*}$, Sifan Wang$^{2}$, Leevan Ling$^{1}$ \\[4pt]
{\small $^{1}$Department of Mathematics, Hong Kong Baptist University, Hong Kong SAR, China} \\
{\small $^{2}$Institution for Foundations of Data Science, Yale University, New Haven, CT 06520, USA}\\
{\small *Corresponding author: \texttt{amir\_noori@hkbu.edu.hk}}
}

\begin{document}
\maketitle

\begin{abstract}
\noindent
Kolmogorov--Arnold Networks (KANs) replace fixed activations with learnable univariate edge functions whose behavior depends strongly on the chosen basis. Gaussian radial basis functions provide a simple and efficient alternative to splines, but their accuracy and stability are highly sensitive to the scale parameter \(\epsilon\), which has not been studied systematically. We analyze this dependence through the geometry and conditioning of the first-layer feature matrix. Because the first layer is defined directly on the input domain, any loss of feature distinguishability introduced there propagates through the entire network. Based on this analysis, we identify the practical operating interval
\[
\epsilon \in \left[\frac{1}{G-1},\frac{2}{G-1}\right],
\]
where \(G\) is the number of Gaussian centers. This interval is proposed as a stable design rule rather than a universal optimum. Extensive experiments on function approximation and physics-informed problems confirm its reliability across different collocation densities, grid resolutions, architectures, and input dimensions. We also show how the same principle supports fixed-scale selection, variable-scale models, constrained optimization of \(\epsilon\), and efficient scale search using early-stage training error. These results establish scale selection as a central design principle for reliable and accurate Gaussian KANs. The code and implementation details are available at \url{https://github.com/AmirNoori68/Gaussian-KAN}.
\end{abstract}


\section{Introduction}
Kolmogorov--Arnold Networks have recently emerged as a structured alternative to classical multilayer perceptrons (MLPs), replacing fixed node activations with learnable univariate functions placed on edges~\cite{Liu24,Liu24b}. 
This edge-based construction provides a different inductive bias for multivariate approximation and has attracted growing interest in scientific machine learning \cite{SS24b,SS24c,SS24d} and physics-informed modeling~\cite{Raissi19,Karniadakis21, Noorizadegan25}. 
The original KAN formulation is built from B-spline bases~\cite{Liu24}, and a rapidly expanding literature has explored alternative basis families to improve smoothness, computational efficiency, and approximation quality.

Recent KAN variants have employed many different basis functions, including Chebyshev-type constructions~\cite{Daryakenari25,Toscano24_kkan}, Jacobi and rational Jacobi bases~\cite{Aghaei24_fkan,Kashefi25}, Fourier-based activations~\cite{Xu25_fourier,Zhang25}, ReLU-based and adaptive piecewise bases~\cite{Qiu24,KAN_pde_So24,pde_Rigas24f}, as well as wavelet, finite-basis, and polynomial variants~\cite{Bozorgasl24,pde_fbkan_Howard24,Seydi24a}. 
In parallel, recent work has also focused on improving adaptability and efficiency through geometric or grid-refinement strategies~\cite{Actor25}, grid-adaptive physics-informed KANs~\cite{pde_Rigas24f}, sparse-identification frameworks~\cite{Howard_2026}, JAX/GPU implementations~\cite{Daryakenari25,pde_Rigas24f}, Feature-Enriched KANs (FEKAN)~\cite{SS26}, operator-learning KANs~\cite{SS26_2}, and improved initialization strategies~\cite{Rigas25_init}.

Among these basis families, Gaussian radial basis functions (RBFs) are particularly attractive because of their smoothness, locality, and simple analytical form. More broadly, Gaussian KANs are part of the long development of RBF methods. Starting from Hardy’s multiquadric formulation~\cite{Hardy71} and Kansa’s extension to PDEs~\cite{Kansa90}, RBF methods have developed into a large literature on approximation\cite{Ling06a} and scientific computing~\cite{Sabir26a, Sabir26b,Sabir26c,Sabir26d,Noorizadegan24, Tizian4}.

In the KAN setting, replacing B-spline activations with Gaussian bases was first proposed in FastKAN~\cite{Li24}, where Gaussian functions were introduced as computationally efficient surrogates for spline-based activations. In that work, the scale parameter was fixed, mainly to show that Gaussian bases can approximate spline constructions while reducing implementation complexity and computational cost. Subsequent works extended Gaussian-based KANs in several directions, including residual formulations~\cite{Koenig25}, operator-learning architectures with learnable Gaussian bases~\cite{Abueidda25}, improved training efficiency and accuracy~\cite{Chiu2026,Athanasios2024}, and complex-valued Gaussian KANs~\cite{Wolff25,Che25_complexKAN}. Beyond regression and PDE-related applications, Gaussian-based KANs have also shown promising performance in classification, with reduced parameter counts, faster training, and improved stability in large-scale or high-dimensional settings~\cite{Chao26}. Taken together, these works show that Gaussian bases are not merely substitutes for splines, but rather constitute a flexible and expressive class of KAN representations.

What remains largely unresolved, however, is how to choose the Gaussian scale parameter \(\epsilon\) in a neural-network setting, or more generally the scale in RBFs that admit such a parameter, since not all RBFs do. In most existing Gaussian-KAN works, \(\epsilon\) is either fixed heuristically~\cite{Li24,Farea26}, selected by trial and error, or made trainable without a clear understanding of its role~\cite{Wen26,Chao26,Chiu2026}. A recent comparative physics-informed study also included Gaussian RBF KANs among several learnable basis families, but treated the Gaussian scale as a fixed preset hyperparameter rather than the subject of a dedicated analysis~\cite{Farea26}. In that study, the authors further noted that RBFs may suffer from low accuracy and vanishing gradients when the scaling factor is not chosen appropriately, again highlighting that basis scaling should be handled carefully for stable training.

In classical Gaussian RBF interpolation, the scale parameter governs the balance between approximation power and numerical stability, and poor choices can lead either to weak approximation or to severe ill-conditioning~\cite{Schaback23,Larsson24,Roberto19,Roberto23}. This is closely related to the stability--accuracy \emph{trade-off} and the \emph{uncertainty-principle} viewpoint in the RBF literature~\cite{Schaback95,Schaback23,Fasshauer06}, and the general optimal-scaling problem remains open even in classical meshfree methods.

This difficulty in choosing the scale parameter, together with the associated conditioning issues, is one reason why strong-form meshfree RBF methods, despite their potentially much higher accuracy~\cite{Alex03}, have not replaced the more stable weak-form finite element framework. However, as Larsson and Schaback emphasized,
\begin{quote}
\emph{the strong dependence of radial basis function techniques on scaling is a feature, not a bug}~\cite{Larsson24}.
\end{quote}
Rather than treating scale sensitivity as a defect, it can be viewed as a mechanism that, when understood properly, allows one to control approximation behavior.

Noorizadegan et al.~\cite{Amir22,AmirQC}, inspired by Schaback’s conditioning--accuracy trade-off, showed that working within a numerically safe regime, while pushing the conditioning as high as safely allowed by the implementation and the available reliable digits~\cite{Amir_eval}, can lead to good accuracy. One step further in this direction is the flat-limit regime~\cite{Larsson05,Fasshauer12}. We do not consider that regime here, and instead keep the analysis within the simple Gaussian-KAN setting used in current practice.

Our main contributions are summarized as follows:

\begin{itemize}
    \item To the best of our knowledge, this work provides the first systematic study devoted specifically to the scale parameter in Gaussian KANs. We show that the scale parameter can change the model behavior from poor approximation to highly accurate approximation and that, when selected properly, Gaussian KANs can achieve accuracy comparable to or better than Chebyshev-based KANs in several test cases.

    \item We formulate the first layer of Gaussian KANs from a kernel-feature viewpoint. This formulation connects Gaussian KANs with classical RBF and meshfree approximation theory and enables practical conditioning diagnostics for understanding and selecting the scale parameter.

    \item We identify the first layer as the critical layer for scale-parameter selection. The analysis shows that the conditioning of the first-layer feature matrix determines whether the Gaussian basis provides a stable and informative representation before the data are propagated through the remaining network.

    \item We identify the practical operating interval
    \[
    \epsilon \in \left[\frac{1}{G-1},\frac{2}{G-1}\right],
    \]
    where \(G\) is the number of Gaussian centers. This interval provides a simple rule for selecting the scale parameter and remains effective across different target functions, training densities, grid sizes, architectures, input dimensions, and physics-informed PDE problems.

    \item We propose two practical strategies for selecting the scale parameter when the exact solution is unavailable: a variable shape-parameter approach and an efficient search procedure based on the proposed conditioning analysis. These strategies reduce the need for exhaustive trial-and-error sweeps.

    \item We show that the same conditioning-based viewpoint can be extended beyond Gaussian KANs to other RBF-based KANs, including Mat\'ern KANs.
\end{itemize}

Together, these contributions address an important practical gap in Gaussian KANs. The scale parameter has often been selected by trial and error, which can make Gaussian bases appear unreliable even when the poor performance is mainly caused by an inappropriate scale choice.

The remainder of this paper is organized as follows. Section~\ref{sec:formulation} introduces the motivation from Kolmogorov superposition and presents the Gaussian KAN formulation, together with the experimental setup and benchmark target functions. Section~\ref{sec:first_layer} develops the first-layer kernel viewpoint, establishes the first layer as the main representational bottleneck, and supports this role with layer-wise numerical evidence. Section~\ref{sec:conditioning_interval} studies first-layer conditioning and derives a practical interval for the Gaussian scale parameter, then examines how this interval behaves under changes in collocation density, number of Gaussian centers, network architecture, and input dimension, and finally discusses shared variable scales, training-MSE-based scale search, and a representative physics-informed Helmholtz and  and digital option problem. Section~\ref{sec:conclusion} concludes the paper, and Section~\ref{sec:limitations} discusses limitations and directions for future work.

\section{From Kolmogorov Superposition to Gaussian KAN Formulation}
\label{sec:formulation}

In a Gaussian KAN, the scale parameter \(\epsilon\) enters through the univariate edge features. 
Therefore, before studying conditioning, rank loss, or representation collapse, we first write the KAN layer explicitly as a finite Gaussian feature map followed by a linear coefficient map. 
This formulation will make clear why the first-layer feature matrix is the central object in our analysis.

\subsection{Kolmogorov Superposition as Architectural Motivation}

The Kolmogorov superposition theorem provides the structural motivation behind KANs: multivariate dependence can be represented through sums and compositions of univariate functions. 
In its classical form, the theorem states the following.

\begin{theorem}[Kolmogorov superposition theorem~\cite{Kolmogorov57}]
\label{thm:kst}
Let \(d\ge 2\) and let \(C([0,1]^d)\) denote the space of continuous real-valued functions on \([0,1]^d\). 
Then there exist continuous functions \(\theta_{q,p}:[0,1]\to\mathbb{R}\), independent of the target function \(f\), such that every \(f\in C([0,1]^d)\) admits a representation of the form
\begin{equation}
f(x_1,\dots,x_d)
=
\sum_{q=0}^{2d}
\Psi_q\!\left(
\sum_{p=1}^{d}\theta_{q,p}(x_p)
\right),
\label{eq:kst_representation}
\end{equation}
where the outer functions \(\Psi_q:\mathbb{R}\to\mathbb{R}\) are continuous and depend on \(f\).
\end{theorem}

For the present work, the important point is not the exact constructive content of Theorem~\ref{thm:kst}, but the architectural principle suggested by \eqref{eq:kst_representation}: coordinate-wise univariate transformations can be combined to represent multivariate functions. 
KANs adopt this principle by placing trainable univariate functions on network edges and summing their outputs at each node. 
Thus, KANs should be viewed as architectures inspired by the Kolmogorov--Arnold structure, rather than as direct numerical realizations of the theorem. 
A more detailed discussion of the relation between KST and KANs can be found in~\cite{AmirKAN}.

This viewpoint is directly relevant to the present paper because the Gaussian scale parameter \(\epsilon\) controls the shape of these univariate edge functions. 
Consequently, scale selection is not merely a local basis-function issue; it affects how the network constructs and propagates representations through its layers.

\subsection{Gaussian KAN Layer Formulation}

Consider a layer \(\ell\) with input width \(n_\ell\) and output width \(n_{\ell+1}\). 
Given an input vector \(x^{(\ell)}\in\mathbb{R}^{n_\ell}\), a KAN layer assigns a univariate function
\[
\psi^{(\ell)}_{j,i}:\mathbb{R}\to\mathbb{R}
\]
to each edge from input coordinate \(i\) to output coordinate \(j\). 
The output \(x^{(\ell+1)}\in\mathbb{R}^{n_{\ell+1}}\) is defined componentwise by
\begin{equation}
x^{(\ell+1)}_j
=
\sum_{i=1}^{n_\ell}
\psi^{(\ell)}_{j,i}\!\left(x^{(\ell)}_i\right),
\qquad
j=1,\dots,n_{\ell+1}.
\label{eq:kan_layer_definition}
\end{equation}
This edge-based structure is the part of the KAN architecture that matters for our analysis: each coordinate is first transformed by a univariate basis expansion, and only afterward are the resulting contributions combined.

In the Gaussian KAN studied here, every edge function in \eqref{eq:kan_layer_definition} is expanded in a shared finite Gaussian basis. 
We fix a set of centers
\begin{equation}
\mathcal{C}
=
\{c_1,\dots,c_G\}
\subset[0,1],
\label{eq:shared_centers}
\end{equation}
and a shared scale parameter
\begin{equation}
\epsilon>0.
\label{eq:shared_epsilon}
\end{equation}
The centers are placed in \([0,1]\) because the original input coordinates are scaled to this interval. 
Although deeper-layer coordinates are not necessarily constrained to remain in \([0,1]\), the same Gaussian feature map is evaluated as a function on all of \(\mathbb{R}\):
\begin{equation}
\varphi:\mathbb{R}\to\mathbb{R}^G,
\qquad
\varphi(t)
=
\begin{bmatrix}
\exp\!\left(-\dfrac{(t-c_1)^2}{\epsilon^2}\right)\\
\vdots\\
\exp\!\left(-\dfrac{(t-c_G)^2}{\epsilon^2}\right)
\end{bmatrix},
\label{eq:gaussian_feature_map}
\end{equation}
where \(G\) is the number of Gaussian centers.

Each edge function is then written as
\begin{equation}
\psi^{(\ell)}_{j,i}(t)
=
\bigl(w^{(\ell)}_{j,i}\bigr)^\top \varphi(t),
\qquad
w^{(\ell)}_{j,i}\in\mathbb{R}^G.
\label{eq:gaussian_edge_function}
\end{equation}
Thus, the nonlinear part of the layer is entirely determined by the Gaussian feature map \(\varphi\), while the trainable dependence on the data enters linearly through the coefficients \(w^{(\ell)}_{j,i}\). 
This separation is important: the scale parameter \(\epsilon\) changes the geometry of the feature matrix before any coefficient learning takes place.

To write the layer compactly, define the stacked feature vector
\begin{equation}
\Phi\!\left(x^{(\ell)}\right)
=
\begin{bmatrix}
\varphi\!\left(x^{(\ell)}_1\right)\\
\vdots\\
\varphi\!\left(x^{(\ell)}_{n_\ell}\right)
\end{bmatrix}
\in\mathbb{R}^{n_\ell G},
\label{eq:stacked_feature_vector}
\end{equation}
and the block coefficient matrix
\begin{equation}
W^{(\ell)}
=
\begin{bmatrix}
\bigl(w^{(\ell)}_{1,1}\bigr)^\top & \cdots & \bigl(w^{(\ell)}_{1,n_\ell}\bigr)^\top\\
\vdots & \ddots & \vdots\\
\bigl(w^{(\ell)}_{n_{\ell+1},1}\bigr)^\top & \cdots & \bigl(w^{(\ell)}_{n_{\ell+1},n_\ell}\bigr)^\top
\end{bmatrix}
\in\mathbb{R}^{\,n_{\ell+1}\times n_\ell G}.
\label{eq:gaussian_layer_matrix}
\end{equation}
Then the Gaussian KAN layer takes the finite-feature form
\begin{equation}
x^{(\ell+1)}
=
W^{(\ell)}
\Phi\!\left(x^{(\ell)}\right).
\label{eq:gaussian_kan_layer_matrix_form}
\end{equation}

For later use, define the layer operator
\begin{equation}
F^{(\ell)}:\mathbb{R}^{n_\ell}\to\mathbb{R}^{n_{\ell+1}},
\qquad
F^{(\ell)}(z)
=
W^{(\ell)}\Phi(z).
\label{eq:layer_operator_definition}
\end{equation}
A deep Gaussian KAN with \(L\) layers is then written as
\begin{equation}
f
=
F^{(L-1)}
\circ
F^{(L-2)}
\circ
\cdots
\circ
F^{(0)},
\qquad
f:[0,1]^{n_0}\to\mathbb{R}^{n_L}.
\label{eq:deep_gaussian_kan}
\end{equation}

This formulation is the bridge between the KAN architecture and the conditioning analysis that follows. 
Every layer has the same algebraic structure: a Gaussian feature map controlled by \(\epsilon\), followed by a linear coefficient map. 
However, the first layer is distinct because its features are evaluated directly on the original input coordinates. 
Therefore, any poor scale choice in the first layer immediately affects the geometry of the input representation before the network has any opportunity to transform the data. 
For this reason, the next section focuses on the first-layer feature map, its induced kernel, and the resulting conditioning behavior.

\section{First-Layer Kernel Structure and Bottleneck Role}\label{sec:first_layer}

\subsection{First-Layer Feature Map and Induced Kernel}

The first layer is the only layer whose Gaussian basis is evaluated directly on the original input domain. 
For an input \(x=(x_1,\dots,x_d)^\top\in[0,1]^d\), the first-layer feature representation is obtained by stacking the one-dimensional Gaussian features coordinate-wise:
\begin{equation}
\Phi(x)
=
\begin{bmatrix}
\varphi(x_1)\\
\vdots\\
\varphi(x_d)
\end{bmatrix}
\in\mathbb{R}^{dG}.
\label{eq:first_layer_stacked_feature}
\end{equation}
The first hidden representation is therefore
\begin{equation}
x^{(1)}
=
W^{(0)}\Phi(x),
\label{eq:first_layer_map}
\end{equation}
where \(W^{(0)}\in\mathbb{R}^{n_1\times dG}\). 
Thus, the first layer is a finite-dimensional feature map followed by a linear transformation.

The associated one-dimensional kernel induced by the Gaussian feature map \eqref{eq:gaussian_feature_map} is
\begin{equation}
k(s,t)
=
\varphi(s)^\top\varphi(t),
\qquad s,t\in[0,1].
\label{eq:first_layer_scalar_kernel}
\end{equation}
Accordingly, the first-layer kernel between two inputs \(x,x'\in[0,1]^d\) is
\begin{equation}
K_0(x,x')
=
\Phi(x)^\top \Phi(x').
\label{eq:first_layer_kernel_definition}
\end{equation}
Using the stacked structure in \eqref{eq:first_layer_stacked_feature}, this becomes
\begin{equation}
K_0(x,x')
=
\sum_{i=1}^{d}\varphi(x_i)^\top\varphi(x_i')
=
\sum_{i=1}^{d} k(x_i,x_i').
\label{eq:first_layer_kernel_sum}
\end{equation}
Hence the first Gaussian KAN layer induces an additive kernel over the input coordinates. 
This representation is exact and directly reflects the edge-wise structure of KANs in \eqref{eq:kan_layer_definition}: each coordinate is processed separately, and the resulting contributions are summed.

For later conditioning analysis, consider a finite sample set
\begin{equation}
\mathcal{X}
=
\{x^{1},\dots,x^{N}\}
\subset[0,1]^d.
\label{eq:sample_set_definition}
\end{equation}
For each coordinate \(i=1,\dots,d\), define the feature block
\begin{equation}
\Phi^{(i)}
=
\begin{bmatrix}
\varphi\!\left(x_i^{1}\right)^\top\\
\vdots\\
\varphi\!\left(x_i^{N}\right)^\top
\end{bmatrix}
\in\mathbb{R}^{N\times G}.
\label{eq:first_layer_coordinate_matrix}
\end{equation}
Stacking these blocks horizontally gives the full first-layer feature matrix
\begin{equation}
\Phi
=
\begin{bmatrix}
\Phi^{(1)} & \cdots & \Phi^{(d)}
\end{bmatrix}
\in\mathbb{R}^{N\times dG}.
\label{eq:full_first_layer_matrix}
\end{equation}
The empirical first-layer kernel matrix is then
\begin{equation}
K_0
=
\Phi\Phi^\top
\in\mathbb{R}^{N\times N},
\label{eq:empirical_first_layer_kernel}
\end{equation}
and, by the block structure of \(\Phi\),
\begin{equation}
K_0
=
\sum_{i=1}^{d}
\Phi^{(i)}\bigl(\Phi^{(i)}\bigr)^\top.
\label{eq:empirical_first_layer_kernel_sum}
\end{equation}
Thus, the empirical first-layer kernel matrix is the sum of the coordinate-wise Gram matrices. 
This matrix is the starting point for the later analysis of rank, conditioning, and admissible Gaussian scales.

\subsection{The First Layer as a Representational Bottleneck}
\label{subsec:first_layer_bottleneck}

Because the first layer acts directly on the original variables, any loss of distinguishability introduced at this stage is inherited by all later layers. 
Writing the full network as
\begin{equation}
f
=
T\circ G,
\qquad
G(x)=W^{(0)}\Phi(x),
\label{eq:network_factorization_first_layer}
\end{equation}
where \(T=F^{(L-1)}\circ\cdots\circ F^{(1)}\) denotes the composition of the remaining layers, makes this dependence explicit. 
The downstream map \(T\) only receives the first hidden representation \(G(x)\), not the original input \(x\).

\begin{proposition}[First-layer bottleneck]
\label{prop:first_layer_bottleneck}
Let \(f\) be the Gaussian KAN written as in \eqref{eq:network_factorization_first_layer}. Then:
\begin{enumerate}
\item[(i)] If two inputs \(x,x'\in[0,1]^d\) satisfy
\begin{equation}
G(x)=G(x'),
\label{eq:same_first_layer_output}
\end{equation}
then
\begin{equation}
f(x)=f(x').
\label{eq:same_final_output}
\end{equation}

\item[(ii)] If two inputs \(x,x'\in[0,1]^d\) satisfy
\begin{equation}
\Phi(x)=\Phi(x'),
\label{eq:same_first_layer_feature}
\end{equation}
then they induce the same first-layer kernel section,
\begin{equation}
K_0(x,z)=K_0(x',z),
\qquad
z\in[0,1]^d,
\label{eq:same_kernel_section}
\end{equation}
and therefore
\begin{equation}
f(x)=f(x')
\label{eq:same_output_from_same_feature}
\end{equation}
for every choice of \(W^{(0)}\) and every downstream map \(T\).

\item[(iii)] On a finite sample set \(\mathcal{X}=\{x_j\}_{j=1}^N\subset[0,1]^d\), as \(\epsilon\to\infty\),
\begin{equation}
\Phi(x)\to \mathbf{1}_{dG}
\qquad
\text{for every }x\in[0,1]^d,
\label{eq:pointwise_feature_collapse}
\end{equation}
and hence
\begin{equation}
\Phi
\to
\mathbf{1}_{N}\mathbf{1}_{dG}^{\top}.
\label{eq:matrix_feature_collapse}
\end{equation}
Consequently,
\begin{equation}
K_0=\Phi\Phi^\top
\to
dG\,\mathbf{1}_{N}\mathbf{1}_{N}^{\top},
\label{eq:kernel_matrix_collapse}
\end{equation}
so the empirical first-layer kernel collapses to a rank-one limit. 
Moreover,
\begin{equation}
G(x)
=
W^{(0)}\Phi(x)
\to
W^{(0)}\mathbf{1}_{dG},
\label{eq:first_hidden_constant_limit}
\end{equation}
which means that the first hidden representation becomes constant across the sample set.
\end{enumerate}
\end{proposition}

The proof of Proposition~\ref{prop:first_layer_bottleneck} is given in Appendix~\ref{app:first_layer_bottleneck_proof}. 
The proposition formalizes the special role of the first layer: once two inputs become indistinguishable after the first-layer feature map, no later layer can separate them. 
In particular, large values of the shared Gaussian scale \(\epsilon\) force the first-layer features toward a rank-one limit, causing collapse of both the empirical kernel matrix and the hidden representation. 
For this reason, the numerical behavior of the Gaussian KAN is governed primarily by the geometry and conditioning of the first-layer feature matrix.

\subsection{Experimental Setup and Benchmark Target Functions}

Unless stated otherwise, we use a Gaussian KAN with architecture
\[
[2,12,12,1],
\]
that is, a two-dimensional input, two hidden layers of width \(12\), and a one-dimensional output. The Gaussian scale parameter \(\epsilon\) is swept over \(100\) logarithmically spaced values in
\[
\epsilon \in [5\times 10^{-3},\,5].
\]
All sweep experiments are trained for \(10000\) epochs using full-batch AdamW with a learning rate of \(10^{-3}\).

The Gaussian centers are fixed and uniformly distributed on \([0,1]\),
\[
c_g=\frac{g-1}{G-1},
\qquad g=1,\ldots,G,
\]
and the trainable coefficients are initialized as
\[
a_{j,g}\sim
\mathcal{N}\!\left(
0,\frac{1}{(d_{\mathrm{in}}G)^2}
\right),
\]
equivalently with standard deviation \(1/(d_{\mathrm{in}}G)\), where \(d_{\mathrm{in}}\) is the input dimension and \(G\) is the number of Gaussian centers. We also tested Xavier-normal initialization,
\[
a_{j,g}\sim
\mathcal{N}\!\left(
0,\frac{2}{d_{\mathrm{in}}G+d_{\mathrm{out}}}
\right),
\]
but the above initialization was more stable in our experiments. Both initializations produced similar scale-dependent behavior, as expected, because the proposed criterion is based on the first-layer feature matrix before the coefficients are applied and is therefore independent of coefficient initialization.

We further compared AdamW and L-BFGS and found similar error trends and useful scale regions. The proposed conditioning-based rule and scale-selection strategies remained effective for both optimizers, since the first-layer feature matrix is constructed before optimization and does not depend on the selected optimizer.

The dense sweep is used as a brute-force reference, allowing the proposed conditioning-based interval to be evaluated against the best empirical choices observed within the tested search range.

Training points are sampled on \([0,1]^2\) using Halton sequences; see Fasshauer~\cite{Fasshauer06} and the discussion in Fasshauer and Iske~\cite{Fasshauer12}. Halton points are low-discrepancy sequences that provide a more uniform coverage of the domain than purely random samples. Such point sets are widely used in radial basis function meshfree methods.

For each experiment, \(4\) or \(5\) random seeds are used, accounting for both the random initialization of the coefficients and the point distribution. 
In the plots, the solid curve shows the geometric mean of the error over the seeds in log scale, while the shaded region indicates the minimum and maximum values across seeds.

As the main error measure, we use the validation root mean square error (RMSE), defined by
\begin{equation}
\mathrm{RMSE}
=
\left(
\frac{1}{M}
\sum_{i=1}^{M}
\bigl(
\hat{u}(x_i,y_i)-u(x_i,y_i)
\bigr)^2
\right)^{1/2},
\end{equation}
where \(\{(x_i,y_i)\}_{i=1}^M\) are the validation points, \(u\) is the target function, and \(\hat{u}\) is the network prediction.

We consider four benchmark target functions, denoted by F1, F2, F3, and F4. 
They were chosen to represent smooth localized structure, periodic oscillation, sharp variation, and discontinuity. 
The functions F1, F2, and F4 are defined on \((x,y)\in[0,1]^2\), whereas F3 is defined on \((x,y)\in[-1,1]^2\).

\begin{itemize}
    \item F1: smooth multi-bump function
    \begin{equation}
    \begin{aligned}
    F1(x,y)
    ={}&0.75\exp\!\left(-\frac{(9x-2)^2+(9y-2)^2}{4}\right)
    +0.75\exp\!\left(-\frac{(9x+1)^2}{49}-\frac{9y+1}{10}\right)
    \\
    &+0.5\exp\!\left(-\frac{(9x-7)^2+(9y-3)^2}{4}\right)
    -0.2\exp\!\left(-\big((9x-4)^2+(9y-7)^2\big)\right).
    \end{aligned}
    \end{equation}

    \item F2: periodic oscillatory function
    \begin{equation}
    F2(x,y)=\sin(4\pi x)\sin(4\pi y).
    \end{equation}

    \item F3: smooth function with a sharp localized peak
    \begin{equation}
    F3(x,y)=\frac{1}{1+10^3(x^2-0.25)^2(y^2-0.25)^2},
    \qquad (x,y)\in[-1,1]^2.
    \end{equation}

    \item F4: discontinuous piecewise oscillatory function
    \begin{equation}
    F4(x,y)=g(x)\left(1+0.15\sin(2\pi y)\right),
    \end{equation}
    where
    \begin{equation}
    g(x)=
    \begin{cases}
    5+\displaystyle\sum_{k=1}^{4}\sin(2\pi kx), & x<0.5,\\[6pt]
    \cos(20\pi x), & x\ge 0.5.
    \end{cases}
    \end{equation}
\end{itemize}

Figure~\ref{target_functions} shows the surfaces of these four target functions.

\begin{figure}[t]
    \centering
    \includegraphics[width=7in]{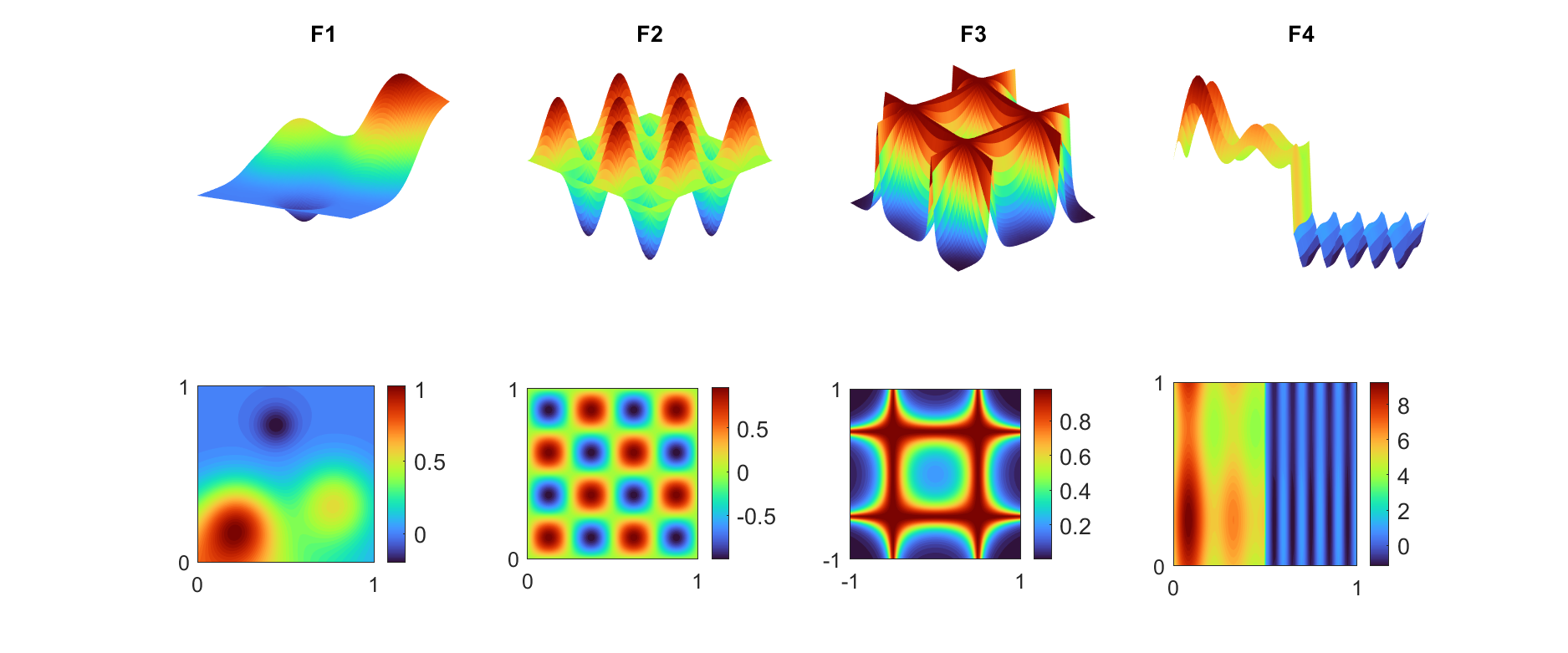}
    \caption{Three-dimensional surfaces of the four target functions used in the regression experiments: F1 (smooth multi-bump), F2 (periodic oscillatory), F3 (sharp localized peak), and F4 (discontinuous piecewise oscillatory). These examples span different approximation challenges, including smooth variation, repeated oscillation, strong localization, and discontinuity.}
    \label{target_functions}
\end{figure}

\subsection{Numerical Verification of First-Layer Dominance}

We now provide numerical evidence for Proposition~\ref{prop:first_layer_bottleneck}.

\paragraph{(i) Layer-wise sensitivity to the Gaussian scale.}
We first vary the Gaussian scale in one layer at a time for the architecture
\[
[2,12,12,1],
\]
which contains three Gaussian mappings,
\[
2\to12,\qquad 12\to12,\qquad 12\to1.
\]
Their scale parameters are denoted by
\begin{equation}
\boldsymbol{\epsilon}
=
(\epsilon_1,\epsilon_2,\epsilon_3).
\label{eq:layerwise_epsilon_vector}
\end{equation}

The non-varied layers are fixed at \(\epsilon=0.1\), and the following cases are compared:
\begin{equation}
(\epsilon,0.1,0.1),\qquad
(0.1,\epsilon,0.1),\qquad
(0.1,0.1,\epsilon),\qquad
(\epsilon,\epsilon,\epsilon).
\label{eq:layerwise_epsilon_protocol}
\end{equation}
The last case uses a shared scale in all layers.

Figure~\ref{fig:layerwise_sensitivity} reports the validation RMSE for the smooth target \(F1\) and the discontinuous target \(F4\). In both cases, varying the first-layer scale produces an error profile close to that of the shared-scale model, whereas varying only the deeper-layer scales has a much weaker effect. This confirms that the first layer carries the dominant sensitivity to \(\epsilon\). The results also show that one properly selected shared scale is sufficient, since separate layer-wise scales do not provide a clear accuracy improvement.

\paragraph{(ii) Conditioning and controlled layer collapse.}
We next consider the deeper architecture
\[
[2,12,12,12,1],
\]
which contains four Gaussian layers. One layer at a time is forced into the large-\(\epsilon\) collapse regime of Proposition~\ref{prop:first_layer_bottleneck}, while all remaining layers use \(\epsilon=0.1\).

The left panel of Figure~\ref{fig:cond_collapse_combined} shows the resulting validation RMSE for \(F1\). Collapsing earlier layers causes substantially greater error, with the largest degradation occurring when the first layer is collapsed. Collapse in later layers is less damaging.

The right panel shows the spectral condition number \(\kappa(\Phi_\ell)\) of each layer in the uncollapsed network. The first layer exhibits a clear well-conditioned region, while the deeper layers remain highly ill-conditioned over most of the sweep. The apparent reduction in their condition numbers beyond approximately \(\epsilon=0.3\) is a \texttt{float32} finite-precision artifact caused by unresolved small singular values. In \texttt{float64}, this artificial drop is delayed or disappears.

Together, Figures~\ref{fig:layerwise_sensitivity} and~\ref{fig:cond_collapse_combined} support Proposition~\ref{prop:first_layer_bottleneck}: the first layer is the main scale-sensitive bottleneck, and representation loss introduced there has the strongest effect on the final accuracy.

\begin{figure}[hbt!]
\centering
\includegraphics[width=6in]{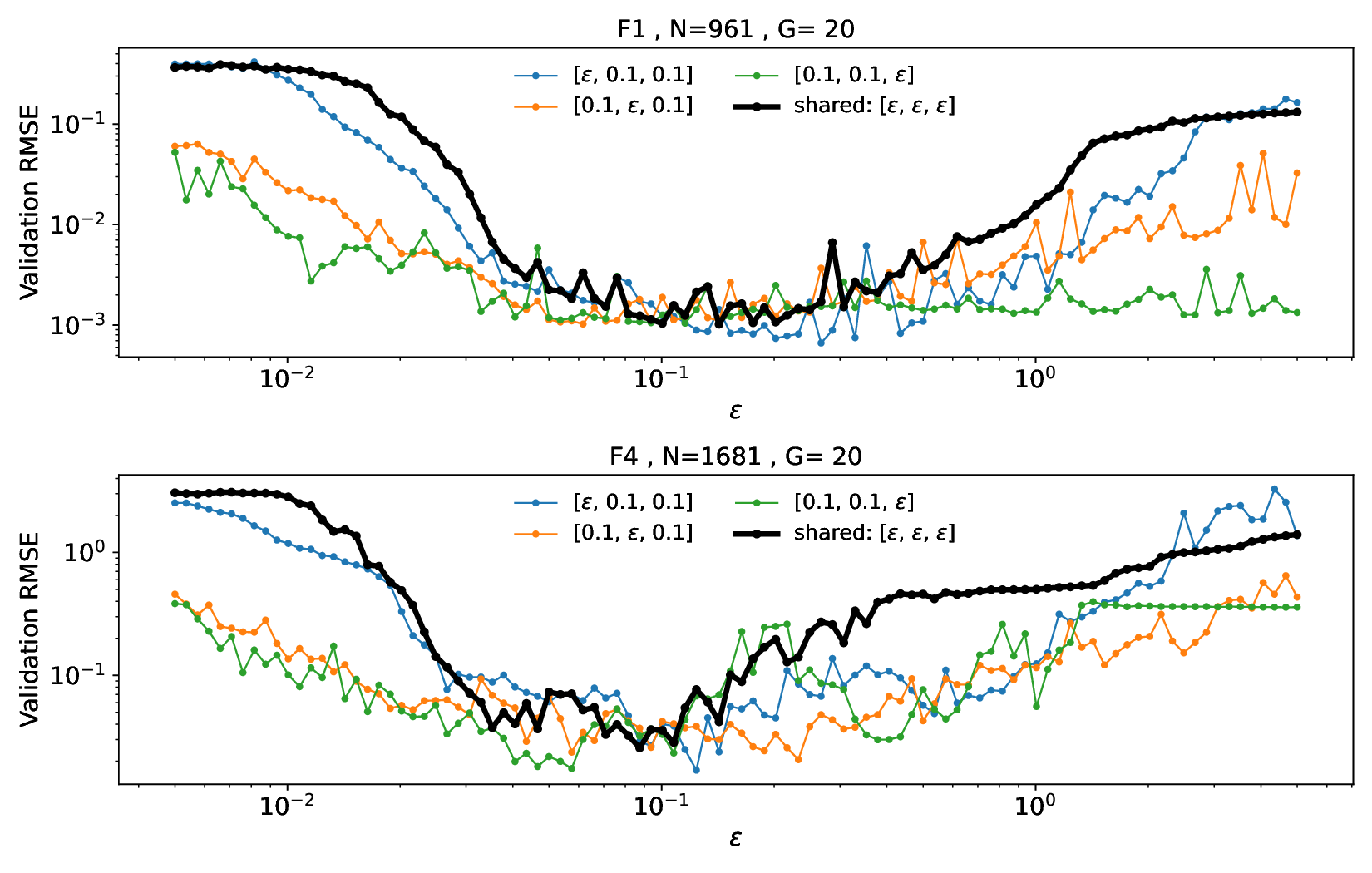}
\caption{
Layer-wise sensitivity of the validation RMSE to the Gaussian scale for the architecture \(2\to 12\to 12\to 1\). 
The reference value is \(\epsilon_0=0.1\), and the four schedules are \((\epsilon,\epsilon_0,\epsilon_0)\), \((\epsilon_0,\epsilon,\epsilon_0)\), \((\epsilon_0,\epsilon_0,\epsilon)\), and \((\epsilon,\epsilon,\epsilon)\). 
For both \(F1\) and \(F4\), varying the first-layer scale \((\epsilon,\epsilon_0,\epsilon_0)\) tracks the all-layer response \((\epsilon,\epsilon,\epsilon)\) much more closely than varying only the deeper layers. This confirms that the first layer carries the dominant sensitivity to the Gaussian scale.
}
\label{fig:layerwise_sensitivity}
\end{figure}

\begin{figure}[hbt!]
\centering
\includegraphics[width=7.5in]{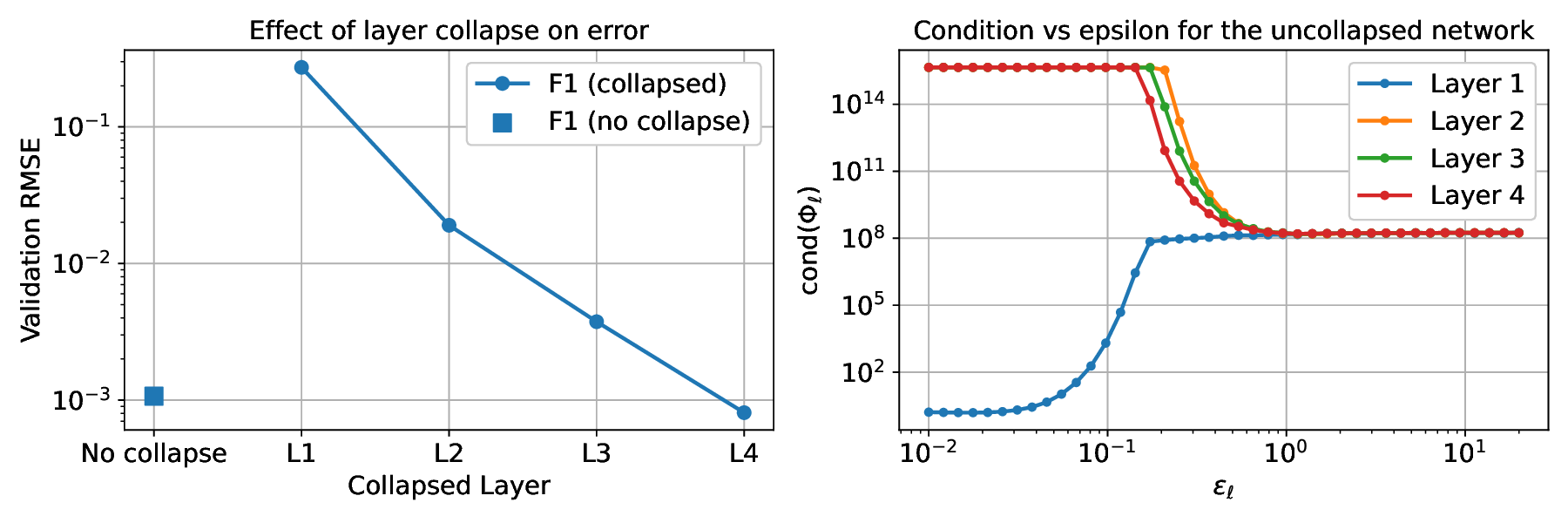}
\caption{
Numerical evidence for first-layer dominance for the architecture \([2,12,12,12,1]\). 
\textbf{Left:} validation RMSE after forcing one layer at a time into the large-\(\epsilon\) collapse regime, while keeping the remaining layers fixed at \(\epsilon=0.1\). The no-collapse case corresponds to using \(\epsilon=0.1\) in all layers. Collapsing earlier layers, especially the first one, causes the largest increase in error. 
\textbf{Right:} spectral condition number \(\kappa(\Phi_\ell)\) of the feature matrix in each layer versus \(\epsilon\) for the uncollapsed network. The first layer is the only one with a clear low-conditioning regime.
}
\label{fig:cond_collapse_combined}
\end{figure}

\section{First-Layer Conditioning and a Practical Gaussian Scale Interval}\label{sec:conditioning_interval}

The goal of this section is to identify a numerically stable operating regime for the first-layer Gaussian feature construction and to translate that regime into a practical rule for choosing the Gaussian scale parameter. 
Since the empirical first-layer kernel matrix in \eqref{eq:empirical_first_layer_kernel} is generated by the pre-weighting feature matrix
\[
\Phi \in \mathbb{R}^{N\times dG},
\]
defined in \eqref{eq:full_first_layer_matrix}, the relevant object for stability analysis is the feature matrix \(\Phi\) itself, or equivalently its coordinate-wise blocks \(\Phi^{(i)}\), rather than the learned coefficient matrix applied afterward.

\subsection{Conditioning Diagnostics for the First-Layer Feature Matrix}
\label{subsec:first_layer_conditioning_diagnostics}

To see how the Gram matrix arises, consider one output component and let \(y\in\mathbb{R}^N\) denote the target values on the sample set \(\mathcal{X}\). 
For a fixed first-layer feature matrix \(\Phi\), fitting coefficients by least squares gives
\begin{equation}
L(w)
=
\frac{1}{2}\|\Phi w-y\|_2^2,
\qquad
w\in\mathbb{R}^{dG}.
\label{eq:least_squares_objective}
\end{equation}
The corresponding normal equations are
\begin{equation}
\Phi^\top \Phi\, w
=
\Phi^\top y,
\label{eq:normal_equations}
\end{equation}
so the Gram matrix \(\Phi^\top\Phi\) appears naturally in coefficient recovery. 
However, column dependence, numerical rank loss, and basis degeneracy originate in \(\Phi\) before they are inherited and amplified by \(\Phi^\top\Phi\).

Let
\begin{equation}
\sigma_1(\Phi)\ge \sigma_2(\Phi)\ge \cdots \ge \sigma_r(\Phi)>0
\label{eq:singular_values_phi}
\end{equation}
denote the positive singular values of \(\Phi\), where \(r=\operatorname{rank}(\Phi)\). 
We define the spectral condition number by
\begin{equation}
\kappa(\Phi)
=
\frac{\sigma_{\max}(\Phi)}{\sigma_{\min}(\Phi)}
=
\frac{\sigma_1(\Phi)}{\sigma_r(\Phi)},
\label{eq:condition_number_phi}
\end{equation}
when \(\Phi\) has full column rank, and set \(\kappa(\Phi)=\infty\) otherwise. 
Similarly,
\begin{equation}
\kappa(\Phi^\top\Phi)
=
\frac{\lambda_{\max}(\Phi^\top\Phi)}
{\lambda_{\min}(\Phi^\top\Phi)}
\label{eq:condition_number_gram}
\end{equation}
when \(\Phi^\top\Phi\) is nonsingular, and \(\kappa(\Phi^\top\Phi)=\infty\) otherwise.

\begin{proposition}[Relation between feature-matrix and Gram-matrix conditioning]
\label{prop:gram_vs_feature_conditioning}
Let \(\Phi\in\mathbb{R}^{N\times M}\), with \(M=dG\). Then:
\begin{enumerate}
\item[(i)] If \(\Phi\) has full column rank, then
\begin{equation}
\kappa(\Phi^\top\Phi)
=
\kappa(\Phi)^2.
\label{eq:gram_feature_condition_square}
\end{equation}

\item[(ii)] If \(\Phi\) is rank deficient, then \(\Phi^\top\Phi\) is singular and
\begin{equation}
\kappa(\Phi)
=
\kappa(\Phi^\top\Phi)
=
\infty.
\label{eq:gram_feature_condition_infinite}
\end{equation}
\end{enumerate}
\end{proposition}

The proof of Proposition~\ref{prop:gram_vs_feature_conditioning} is given in Appendix~\ref{app:gram_feature_conditioning_proof}. 
The proposition shows that the Gram matrix does not introduce a new instability mechanism; it squares the instability already present in the feature matrix. 
For this reason, the primary diagnostic for the first-layer feature construction is \(\kappa(\Phi)\), not \(\kappa(\Phi^\top\Phi)\).

The numerical quality of the first-layer basis is determined by the extreme singular values
\begin{equation}
\sigma_{\max}(\Phi)=\sigma_1(\Phi),
\qquad
\sigma_{\min}(\Phi)=\sigma_r(\Phi).
\label{eq:extreme_singular_values}
\end{equation}
To assess numerical rank, we use the standard floating-point tolerance
\begin{equation}
\mathrm{tol}
=
\max(N,M)\,\sigma_{\max}(\Phi)\,\varepsilon_{\mathrm{mach}},
\qquad
M=dG.
\label{eq:numerical_rank_tolerance}
\end{equation}
Accordingly, the first-layer feature matrix is numerically full rank when
\begin{equation}
\sigma_{\min}(\Phi)>\mathrm{tol}.
\label{eq:numerical_full_rank_condition}
\end{equation}

In float32 arithmetic, a practical additional requirement is that the feature matrix should remain moderately conditioned. 
Motivated by Proposition~\ref{prop:gram_vs_feature_conditioning}, we impose
\begin{equation}
\kappa(\Phi)<3\times 10^3.
\label{eq:practical_kappa_threshold}
\end{equation}
This corresponds to
\begin{equation}
\kappa(\Phi^\top\Phi)<9\times 10^6,
\label{eq:gram_threshold_float32}
\end{equation}
which remains within a conservative numerical range for float32-based training.

Combining \eqref{eq:numerical_full_rank_condition} and \eqref{eq:practical_kappa_threshold}, we define the first-layer diagnostic criteria as
\begin{equation}
\sigma_{\min}(\Phi)>\mathrm{tol}
\qquad\text{and}\qquad
\kappa(\Phi)<3\times 10^3.
\label{eq:first_layer_diagnostic_criteria}
\end{equation}
The first condition detects numerical rank loss, while the second controls near-dependence before rank is lost. 
Together, they identify the stable operating regime of the first-layer feature construction. 
These diagnostics may be applied either to the full matrix \(\Phi\) or to the coordinate-wise blocks \(\Phi^{(i)}\), but in all cases the assessment is performed before weighting.

\subsection{Derivation of a Practical Gaussian Scale Interval}

We seek a practical interval of the form
\begin{equation}
\epsilon
\in
\left[
\frac{1}{G-1},
\epsilon_{\mathrm{cond}}(N,G)
\right],
\label{eq:target_interval_form}
\end{equation}
where \(G\) is the number of Gaussian centers per coordinate. The lower bound is the spacing between adjacent uniformly distributed centers,
\begin{equation}
\epsilon_{\min}
=
\frac{1}{G-1}.
\label{eq:eps_spacing_reference}
\end{equation}
Below this scale, the Gaussian basis becomes excessively localized. The upper bound is determined from the conditioning of the first-layer feature matrix \(\Phi(\epsilon)\). As \(\epsilon\) increases, the Gaussian columns become more similar, \(\Phi(\epsilon)\) becomes ill-conditioned, and the representation approaches the collapse described in Proposition~\ref{prop:first_layer_bottleneck}.

We use two markers to locate this transition. The first is the scale at which the condition number is closest to the practical threshold \(3\times10^3\),
\begin{equation}
\epsilon_{\kappa}
=
\arg\min_{\epsilon>0}
\left|
\log \kappa\!\bigl(\Phi(\epsilon)\bigr)
-
\log\!\left(3\times10^3\right)
\right|.
\label{eq:eps_kappa_marker}
\end{equation}
The second is the onset of numerical rank loss,
\begin{equation}
\epsilon_{\mathrm{rank}}
=
\inf\left\{
\epsilon>0:
\sigma_{\min}\!\bigl(\Phi(\epsilon)\bigr)
\leq
\mathrm{tol}(\epsilon)
\right\}.
\label{eq:eps_rank_marker}
\end{equation}

Figure~\ref{fig:Cond_safe} shows that \(\epsilon_{\kappa}\) and \(\epsilon_{\mathrm{rank}}\) are nearly coincident in all four test cases. The minimum validation RMSE also occurs close to this transition. For small \(\epsilon\), insufficient overlap between neighboring Gaussian functions limits the approximation. As \(\epsilon\) increases, the error decreases until the feature matrix approaches the conditioning boundary. Beyond this point, the condition number grows rapidly, numerical rank is lost, and the error deteriorates.

\begin{figure}[hbt!]
\centering
\subfigure[]{\label{eps_diag_F1_N961_G10}
\includegraphics[width=3in]{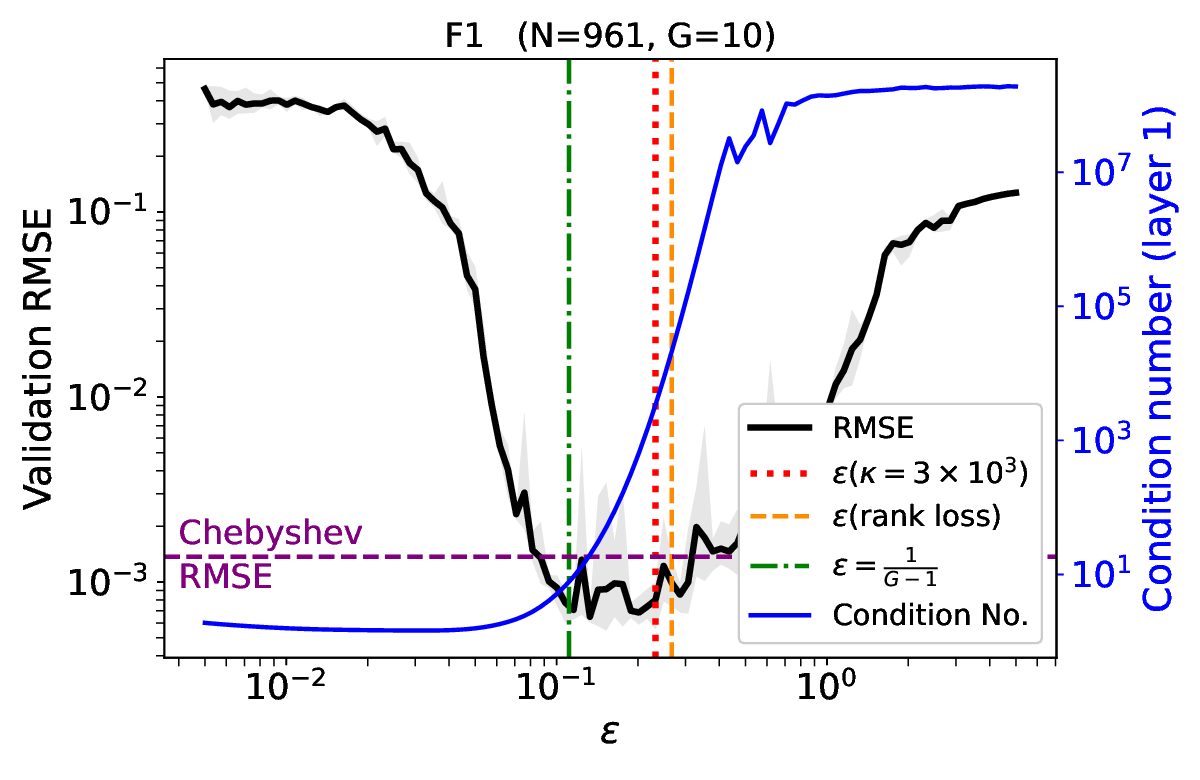}}
\subfigure[]{\label{eps_diag_F11_N961_G14}
\includegraphics[width=3in]{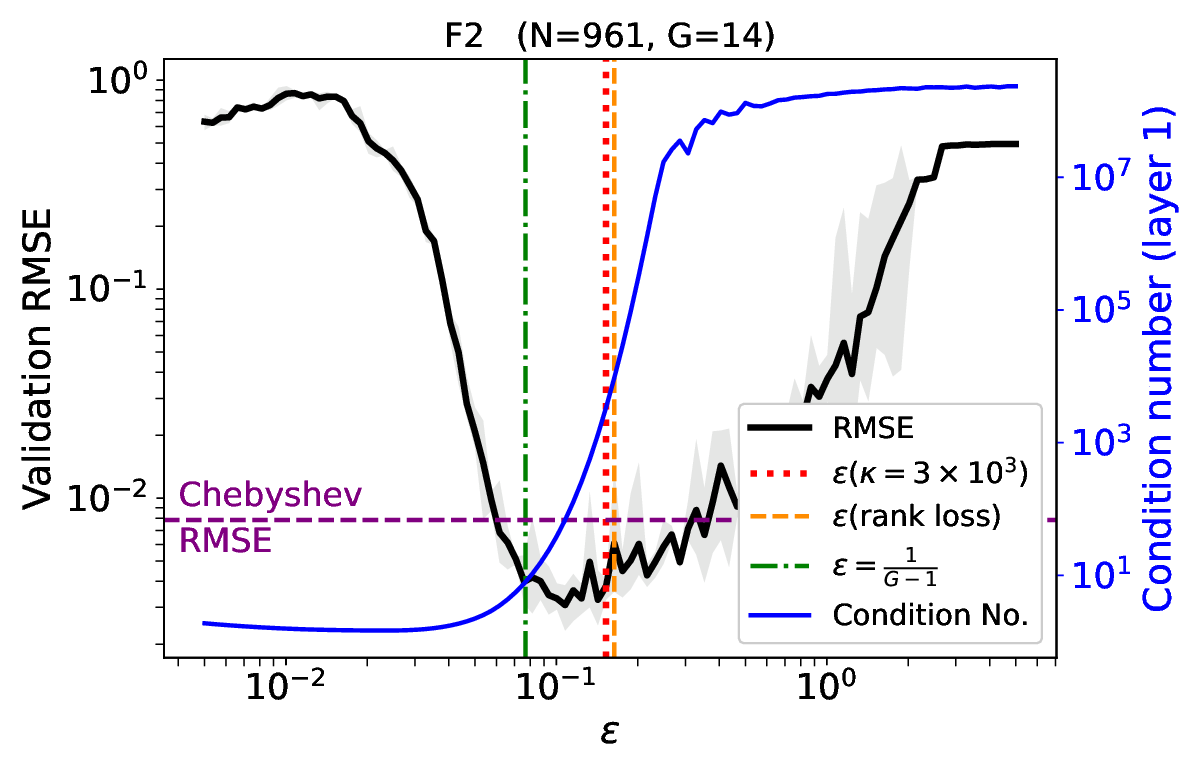}}
\subfigure[]{\label{eps_diag_F20_N1681_G16}
\includegraphics[width=3in]{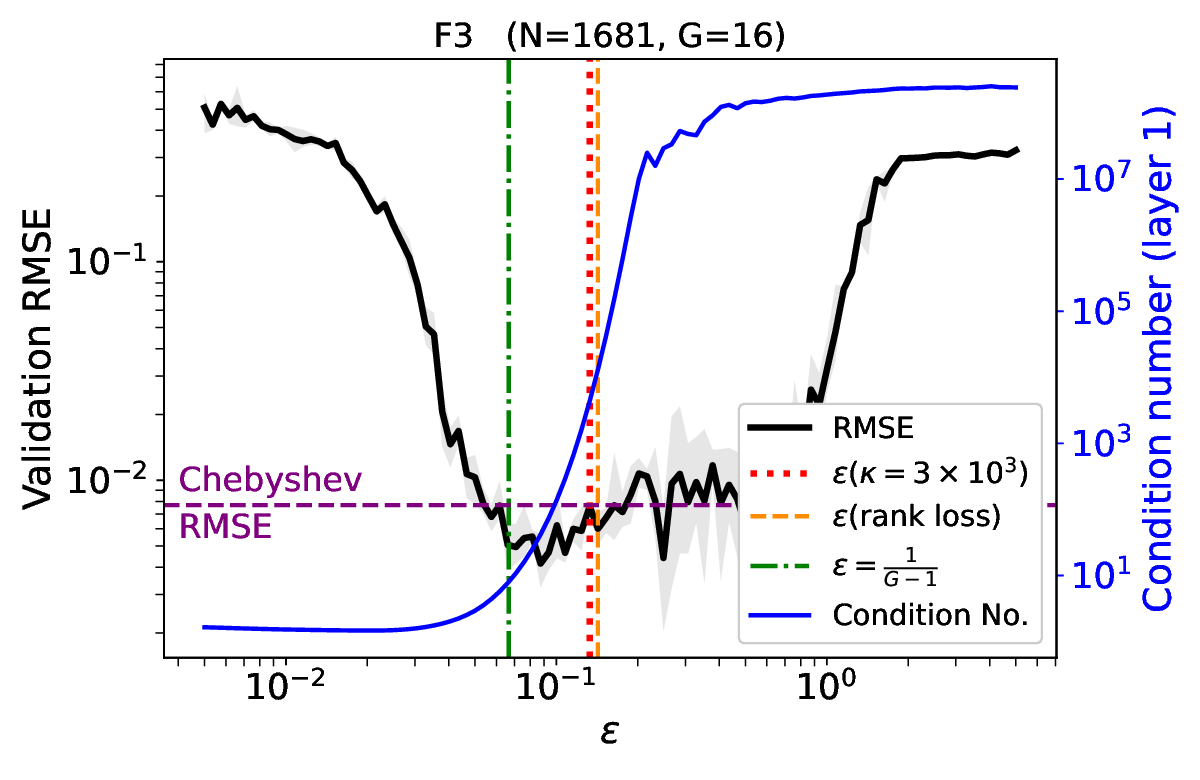}}
\subfigure[]{\label{eps_diag_F24_N1681_G20}
\includegraphics[width=3in]{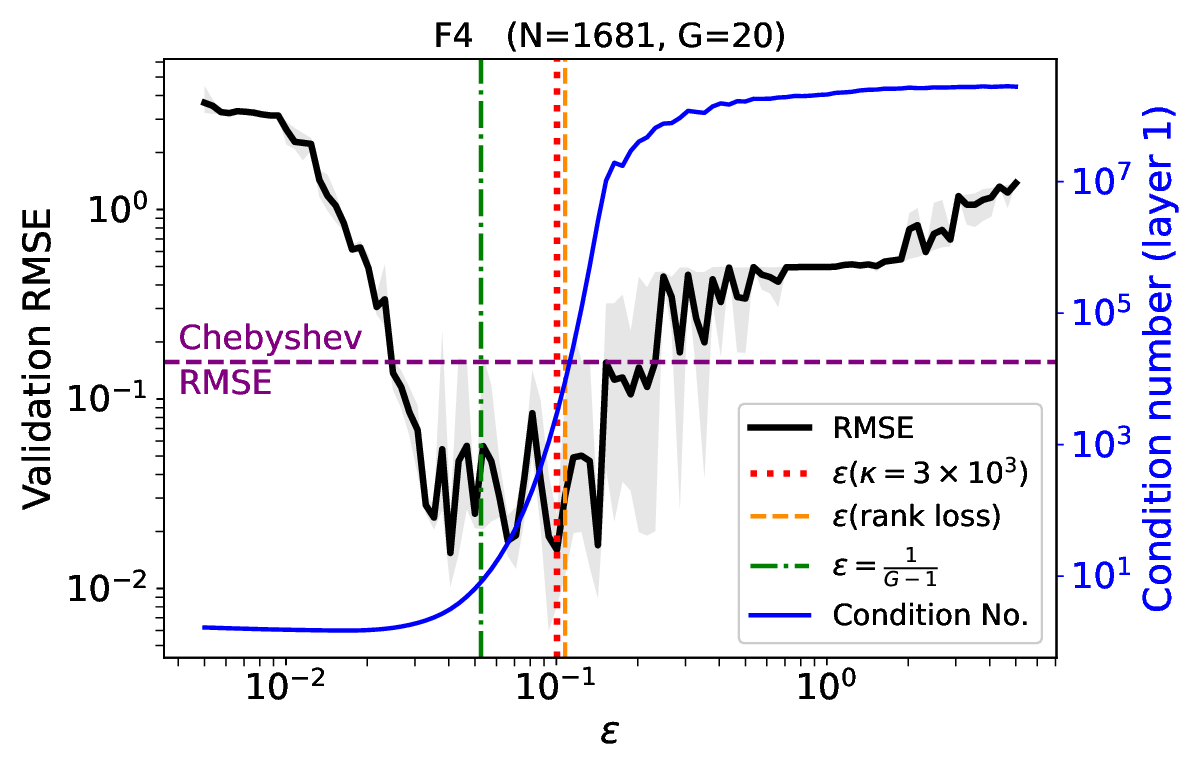}}
\caption{
Validation RMSE and first-layer conditioning as functions of \(\epsilon\). The vertical markers show the center-spacing scale \(\epsilon=1/(G-1)\), the conditioning threshold \(\epsilon_{\kappa}\), and the numerical rank-loss point \(\epsilon_{\mathrm{rank}}\). The latter two are nearly coincident, and the minimum validation RMSE occurs close to this transition. The horizontal purple line shows the validation RMSE of the matched Chebyshev KAN with degree \(m=G-1\); see Paragraph~\ref{par:cheby_reference}.
}
\label{fig:Cond_safe}
\end{figure}

We therefore define
\begin{equation}
\epsilon_{\mathrm{cond}}(N,G)
=
\arg\min_{\epsilon>0}
\left|
\log \kappa\!\bigl(\Phi(\epsilon)\bigr)
-
\log\!\left(3\times10^3\right)
\right|,
\label{eq:epsilon_cond_definition}
\end{equation}
which gives the initial interval
\begin{equation}
\epsilon
\in
\left[
\frac{1}{G-1},
\epsilon_{\mathrm{cond}}(N,G)
\right].
\label{eq:conditioning_interval}
\end{equation}

We next simplify the dependence of \(\epsilon_{\mathrm{cond}}\) on \(N\) and \(G\). For \(N\in[400,3600]\) and \(G\in[6,30]\), \(\epsilon_{\mathrm{cond}}(N,G)\) is computed by sweeping \(\epsilon\) and selecting the value closest to the conditioning threshold. Multiple low-discrepancy sample realizations are used for each pair \((N,G)\), and the resulting scales are combined by a geometric mean. The data are then divided into \(70\%\) training and \(30\%\) validation sets.

We fit the model
\begin{equation}
\log \epsilon_{\mathrm{cond}}
=
\log C
-
q\log(G-1),
\label{eq:log_linear_scaling_model}
\end{equation}
or equivalently,
\begin{equation}
\epsilon_{\mathrm{cond}}
\approx
C(G-1)^{-q}.
\label{eq:power_law_scaling_model}
\end{equation}
The fitted relation is
\begin{equation}
\epsilon_{\mathrm{cond}}
\approx
2.8127\,(G-1)^{-1.1297},
\label{eq:fitted_scaling_law}
\end{equation}
with validation \(R^2\approx0.997\). The fitted exponent is close to \(1\), while the dependence on \(N\) is weak once the samples cover the input interval sufficiently well. This is consistent with the fact that Gaussian overlap is governed mainly by the center spacing.

The fitted relation can therefore be replaced by the simpler rule
\begin{equation}
\epsilon_{\mathrm{cond}}
\approx
\frac{2}{G-1},
\label{eq:simplified_scaling_law}
\end{equation}
which still achieves \(R^2\approx0.985\) on the validation set. Figure~\ref{fig:eps_combined_plot} confirms that this expression closely follows both the measured conditioning scales and their geometric-mean aggregation over \(N\).

We thus obtain the practical interval
\begin{equation}
\boxed{
\epsilon
\in
\left[
\frac{1}{G-1},
\frac{2}{G-1}
\right].
}
\label{eq:practical_epsilon_interval}
\end{equation}
The lower bound ensures sufficient overlap between neighboring Gaussian functions, while the upper bound places the first-layer feature matrix near the observed onset of ill-conditioning and numerical rank loss.

\begin{figure}[hbt!]
\centering
\includegraphics[width=6in]{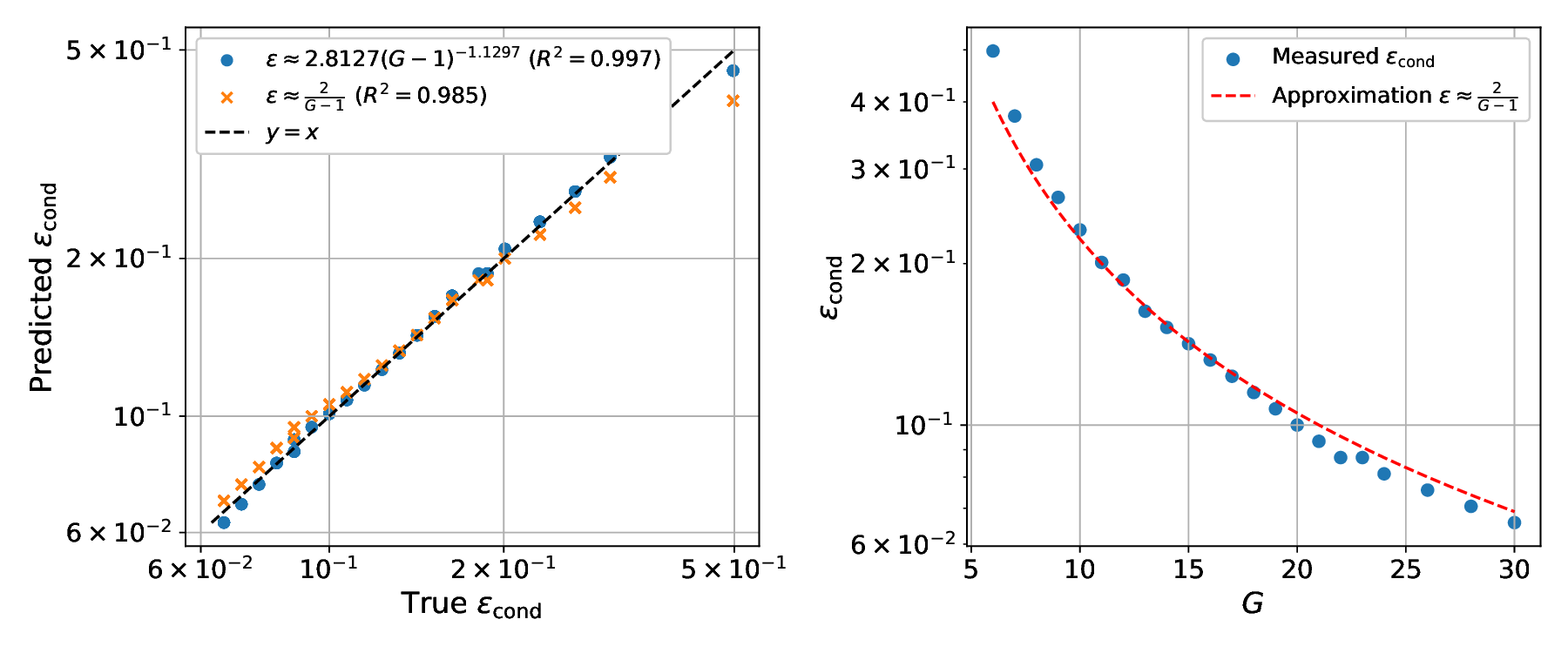}
\caption{
Simplification of the conditioning-based upper bound. \textbf{Left:} predicted versus measured values of \(\epsilon_{\mathrm{cond}}\) on held-out \((N,G)\) pairs. The fitted power law in \eqref{eq:fitted_scaling_law} gives \(R^2\approx0.997\), while the simplified rule in \eqref{eq:simplified_scaling_law} gives \(R^2\approx0.985\). The dashed line denotes exact agreement. \textbf{Right:} geometric-mean values of \(\epsilon_{\mathrm{cond}}\) aggregated over \(N\) for each \(G\), confirming that the dominant dependence is on the number of centers.
}
\label{fig:eps_combined_plot}
\end{figure}

Classical meshfree RBF theory also provides conditioning bounds involving the extreme singular values of the feature matrix~\cite{Fasshauer06}. Although theoretically useful, such bounds are less convenient for routine scale selection. The interval in \eqref{eq:practical_epsilon_interval} is therefore adopted as a simple and directly usable rule.

\paragraph{Chebyshev reference.}\label{par:cheby_reference}
Figure~\ref{fig:Cond_safe} includes a Chebyshev KAN as an accuracy reference. For a fair comparison, the polynomial degree is set to \(m=G-1\), so that both models use \(G\) trainable basis coefficients per edge. The Chebyshev model uses \(\tanh\) between layers to keep the inputs within the admissible basis range, and its final layer also uses a Chebyshev expansion. Under this matched coefficient budget, a properly scaled Gaussian KAN achieves validation errors comparable to, and in several cases lower than, the Chebyshev KAN.

The two bases differ in how accuracy and stability are controlled. In Gaussian KANs, both are adjusted directly through the centers and the scale parameter \(\epsilon\). In Chebyshev KANs, increasing the polynomial degree may improve approximation power but can reduce stability, while lower degrees may be more stable but less accurate; several remedies have therefore been proposed, including basis localization~\cite{Toscano24_kkan}, interlayer normalization~\cite{Yu24}, and architectural modifications~\cite{Daryakenari25}. The Gaussian formulation provides a more direct mechanism because the same scale parameter controls both conditioning and approximation behavior.

This comparison supports the main objective of the present work: showing that Gaussian KANs are competitive when \(\epsilon\) is selected properly. Additional results are reported in Appendix~\ref{app:gaussian_chebyshev_appendix}, where \(G=20\) is fixed and several Chebyshev degrees are tested. These results confirm that poor scale selection can make Gaussian KANs appear weak, whereas a suitable \(\epsilon\) can outperform several Chebyshev configurations. A broader optimized comparison would require a separate study with matched architectures, training procedures, and stabilization strategies.

\subsection{Effect of Collocation Density \texorpdfstring{$N$}{N}}

We now study how the number of collocation points affects approximation in the Gaussian KAN. The goal is to separate two effects: \(N\) controls how well the domain is resolved, while \(\epsilon\) controls the quality and stability of the Gaussian feature representation. We show that increasing \(N\) improves accuracy, but only when \(\epsilon\) is chosen in a suitable range.

All experiments in this subsection are performed on the unit square
\begin{equation}
\Omega=[0,1]^2,
\label{eq:domain_unit_square}
\end{equation}
using Halton points
\begin{equation}
X_N=\{x^1,\dots,x^N\}\subset \Omega
\label{eq:halton_point_set}
\end{equation}
as the training set. Since Halton points are quasi-uniform, their geometric resolution is described by the fill distance
\begin{equation}
h_{X_N,\Omega}
=
\sup_{x\in\Omega}\min_{x^n\in X_N}\|x-x^n\|_2,
\label{eq:fill_distance_definition}
\end{equation}
which in two dimensions satisfies
\begin{equation}
h_{X_N,\Omega}\sim N^{-1/2}.
\label{eq:fill_distance_scaling}
\end{equation}
Thus, increasing \(N\) improves the geometric resolution of the training set.

To generate a systematic refinement sequence, we choose
\begin{equation}
N\in\{k^2:\;k=5,7,9,\dots,87\}.
\label{eq:N_list_definition}
\end{equation}
This gives a monotone progression of collocation densities and makes the scaling in \eqref{eq:fill_distance_scaling} easy to interpret.

Throughout this subsection, the grid size is fixed at \(G=20\). Therefore, the practical interval from \eqref{eq:practical_epsilon_interval} becomes
\begin{equation}
\epsilon\in
\left[
\frac{1}{19},
\frac{2}{19}
\right]
\approx
[0.0526,\,0.1053].
\label{eq:G20_interval}
\end{equation}
We now examine how this interval behaves as \(N\) changes.

\paragraph{Error landscape in \texorpdfstring{$\epsilon$}{epsilon} for different \texorpdfstring{$N$}{N}.}
For fixed architecture and grid size, let
\begin{equation}
E(N,\epsilon)
\label{eq:error_function_N_eps}
\end{equation}
denote the validation RMSE obtained with \(N\) collocation points and Gaussian scale \(\epsilon\). Figure~\ref{fig:error_vs_eps_shadow} plots \(E(N,\epsilon)\) for two representative training sizes, \(N=961\) and \(N=3249\), with the vertical lines marking the interval in \eqref{eq:G20_interval}.

Across all four targets, the same basic structure appears: the error is large for very small \(\epsilon\), decreases into a low-error region, and rises again for large \(\epsilon\). Thus the dependence on \(\epsilon\) remains essentially U-shaped as \(N\) changes. Increasing \(N\) lowers the overall error level, but shifts the useful \(\epsilon\)-window only mildly. In particular, the error-minimizing region remains inside or close to \eqref{eq:G20_interval}. This shows that the conditioning-based interval remains informative even as the collocation density changes.

\paragraph{Best attainable error as a function of \texorpdfstring{$N$}{N}.}
To isolate the effect of collocation density, define
\begin{equation}
E_{\mathrm{opt}}(N)
=
\min_{\epsilon} E(N,\epsilon),
\label{eq:optimal_error_definition}
\end{equation}
that is, the best validation RMSE obtained over the \(\epsilon\)-sweep for a given \(N\). Figure~\ref{fig:optimal_error_vs_N} plots \(E_{\mathrm{opt}}(N)\) against \(N\). The best error decreases monotonically for all four targets, confirming that denser collocation improves approximation quality once the Gaussian scale is chosen well.

The observed decay also reflects target regularity. For the smoother targets F1 and F2, the visible pre-asymptotic behavior is approximately aligned with
\begin{equation}
E_{\mathrm{opt}}(N)\sim h^3 \sim N^{-3/2},
\label{eq:smooth_rate_scaling}
\end{equation}
whereas the less regular targets F3 and F4 are closer to
\begin{equation}
E_{\mathrm{opt}}(N)\sim h^2 \sim N^{-1}.
\label{eq:rough_rate_scaling}
\end{equation}
This is consistent with the standard expectation that smoother targets admit faster convergence, while reduced regularity lowers the observed rate; see Chapter 17 of Fasshauer~\cite{Fasshauer06}.

\paragraph{Interaction between \texorpdfstring{$N$}{N} and \texorpdfstring{$\epsilon$}{epsilon}.}
Figure~\ref{fig:fixed_eps_convergence} shows the validation RMSE versus \(N\) for the fixed Gaussian scales \(\epsilon\in\{0.06,0.08\}\). In all cases, the error decreases as \(N\) increases, but the decay still depends on the choice of \(\epsilon\). Thus, \(N\) and \(\epsilon\) play different roles: \(N\) increases geometric resolution, while \(\epsilon\) controls how effectively the Gaussian feature representation uses that resolution.

\begin{figure}[hbt!]
\centering
\includegraphics[width=7.5in]{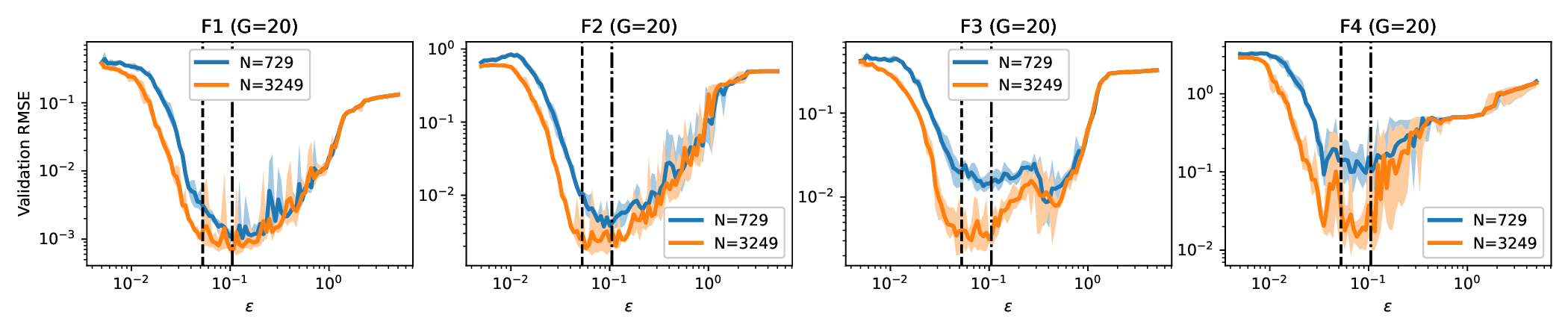}
\caption{
Validation RMSE as a function of the Gaussian scale \(\epsilon\) for two representative collocation densities, \(N=729\) and \(N=3249\), with \(G=20\). The vertical markers indicate \(\epsilon=1/(G-1)\) and \(\epsilon=2/(G-1)\). Across all four targets, increasing \(N\) lowers the error, while the low-error region remains concentrated near the conditioning-based interval.
}
\label{fig:error_vs_eps_shadow}
\end{figure}

\begin{figure}[hbt!]
\centering
\includegraphics[width=7.5in]{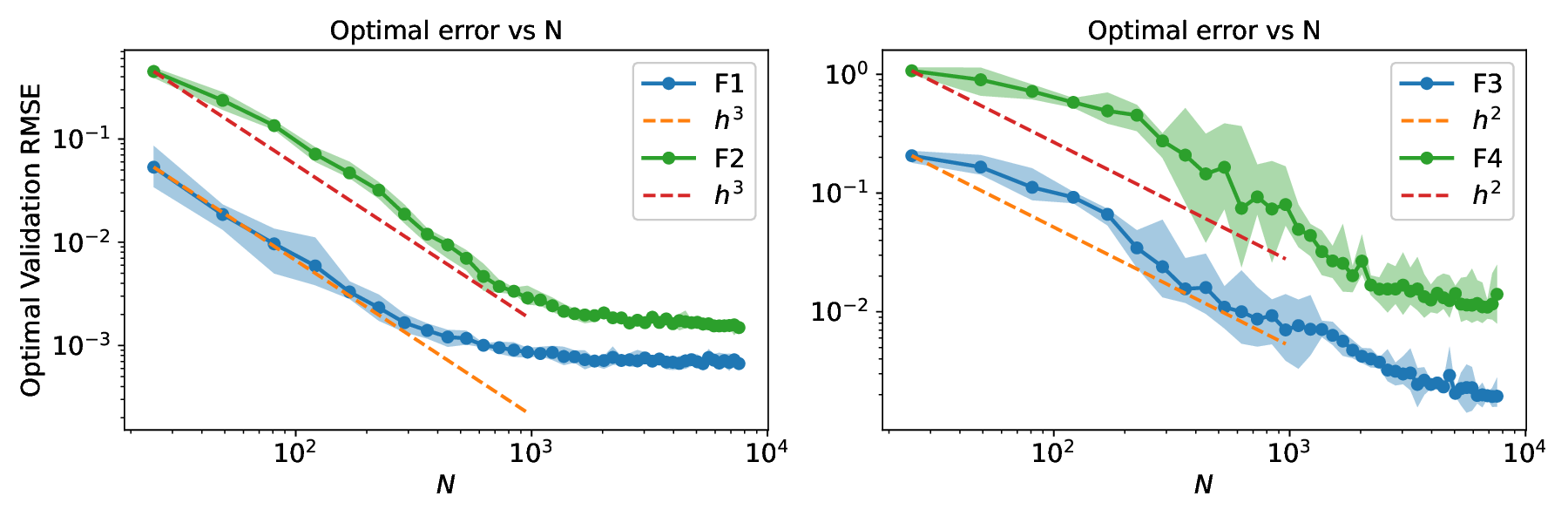}
\caption{
Best validation RMSE \(E_{\mathrm{opt}}(N)\) from \eqref{eq:optimal_error_definition} as a function of the number of collocation points. For the smoother targets F1 and F2, the visible decay is approximately aligned with the \(h^3\) reference, while F3 and F4 are closer to the \(h^2\) reference. In all cases, denser Halton sampling improves the best attainable accuracy.
}
\label{fig:optimal_error_vs_N}
\end{figure}

\begin{figure}[hbt!]
\centering
\includegraphics[width=7.5in]{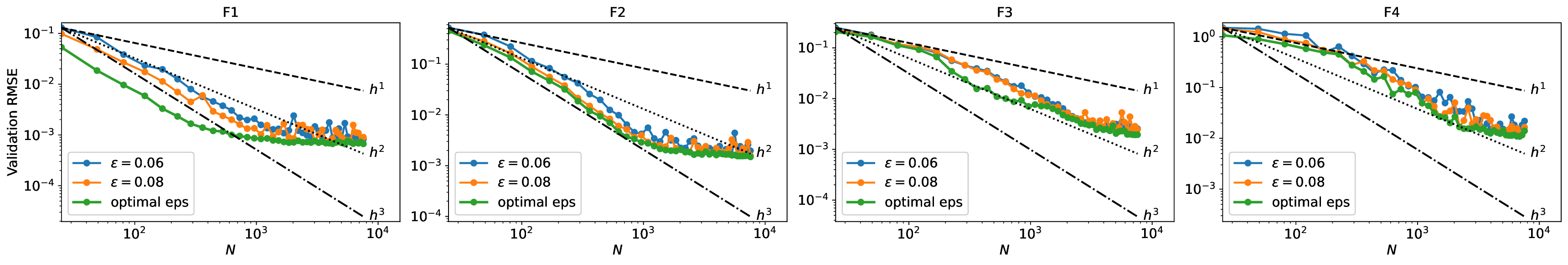}
\caption{
Validation RMSE versus \(N\) for the fixed Gaussian scales \(\epsilon=0.06,0.08\), together with the near-optimal choice at each \(N\) with $G=20$.
}
\label{fig:fixed_eps_convergence}
\end{figure}

\subsection{Effect of the Number of Gaussian Centers \texorpdfstring{$G$}{G}}

We next study the effect of the grid resolution, that is, the number of shared Gaussian centers in each one-dimensional feature map. For fixed architecture and training size \(N\), let
\begin{equation}
\mathcal{C}_G=\{c_1,\dots,c_G\}\subset[0,1]
\label{eq:grid_centers_G}
\end{equation}
denote the uniformly distributed centers. 
The conditioning-based rule in \eqref{eq:simplified_scaling_law} and \eqref{eq:practical_epsilon_interval} already shows that \(G\) and \(\epsilon\) cannot be chosen independently. As \(G\) increases, the center spacing \(\Delta_c\) decreases, so the useful scale range must also move to smaller \(\epsilon\).

\paragraph{Fixed \texorpdfstring{$\epsilon$}{epsilon} versus adaptive \texorpdfstring{$\epsilon$}{epsilon}.}
Figure~\ref{fig:fixedEPS_rmse_vs_G} shows \(E(G,\epsilon)\) for several fixed values of \(\epsilon\), together with a near-optimal choice selected separately for each \(G\). The fixed-\(\epsilon\) curves depend strongly on the chosen scale: a value that works well for one grid size can be clearly suboptimal for another. By contrast, the near-optimal choice gives a much more stable error curve across the full range of \(G\). Thus, refining the grid alone is not enough; \(\epsilon\) must be adjusted with the center spacing.

\paragraph{Shift of the low-error \texorpdfstring{$\epsilon$}{epsilon}-window.}
Figure~\ref{fig:error_vs_eps_G} compares the full \(\epsilon\)-sweeps for two representative grid sizes, \(G=10\) and \(G=32\). Both cases show the same U-shaped dependence on \(\epsilon\), but the effect of increasing \(G\) is different from the effect of increasing \(N\). Increasing \(N\) mainly pushes the curve downward, reducing the error level while leaving the useful \(\epsilon\)-window roughly in the same place. Increasing \(G\), instead, shifts the low-error region to the left, that is, toward smaller \(\epsilon\), and also makes the good-\(\epsilon\) range wider. This is exactly what the proposed range predicts, since that range is tied to the center spacing. 

At the same time, the minimum error itself changes much less than the location of the minimizer. This is consistent with the matrix viewpoint developed earlier: increasing \(G\) mainly changes the geometry of the basis and the position of the stable operating range, rather than acting as a direct accuracy parameter by itself.

\begin{figure}[hbt!]
\centering
\includegraphics[width=7.5in]{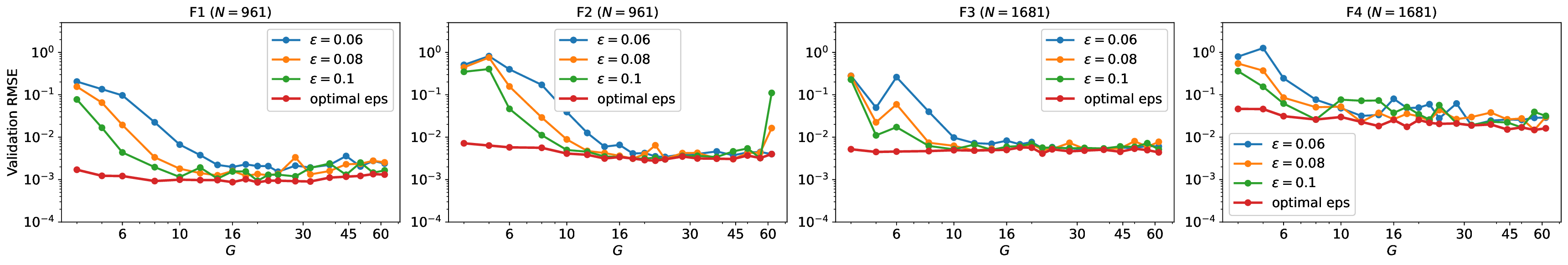}
\caption{
Validation RMSE versus grid resolution \(G\) for several fixed Gaussian scales and for a near-optimal scale chosen separately at each \(G\). The fixed-\(\epsilon\) curves vary strongly with \(G\), whereas the near-optimal choice yields a much more stable error profile.
}
\label{fig:fixedEPS_rmse_vs_G}
\end{figure}

\begin{figure}[hbt!]
\centering
\includegraphics[width=7.5in]{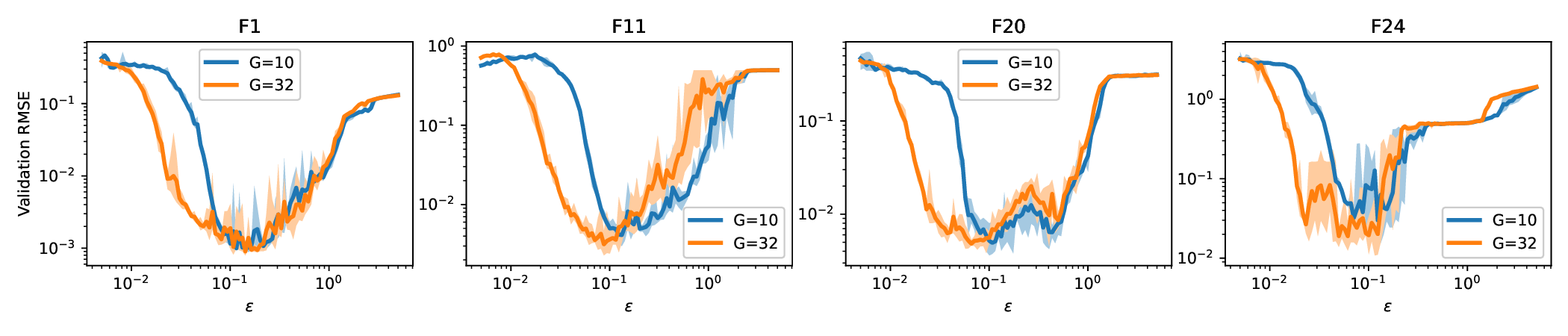}
\caption{
Validation RMSE as a function of the Gaussian scale \(\epsilon\) for two representative grid resolutions, \(G=10\) and \(G=32\). As \(G\) increases, the low-error region shifts to smaller \(\epsilon\) and becomes wider.
}
\label{fig:error_vs_eps_G}
\end{figure}

\subsection{Effect of Network Architecture}

We next study how the network architecture affects the relation between approximation error and the Gaussian scale. Throughout this subsection, the grid resolution is fixed at \(G=20\), so the conditioning-based reference interval is \(\epsilon\in[1/19,\,2/19]\approx[0.0526,\,0.1053]\). Figure~\ref{fig:error_vs_eps_arch} shows that all tested architectures retain the same U-shaped dependence on \(\epsilon\). What changes is the minimum error and the width of the low-error region. Since the first layer is the only layer defined directly on the input domain, it still determines the relevant scale regime. Increasing width or depth does not move the left side of the U-shaped curve, which is the part controlled by the first-layer geometry. Instead, larger architectures mainly enlarge the range of good \(\epsilon\) from the right. This is different from increasing \(G\), which shifts the useful range to the left. The effect is more visible for the less smooth targets F3 and F4. In the top row, the wider architecture \((2,20,20,1)\) generally attains lower error than \((2,4,4,1)\) and makes the low-error region broader. In the bottom row, the deeper architecture \((2,8,8,8,1)\) also lowers the error and broadens the admissible range compared with \((2,8,1)\), again mainly on the larger-\(\epsilon\) side. Thus, a scale range that is suitable for a shallower Gaussian representation remains admissible for a deeper Gaussian KAN; deeper layers may enlarge the useful range, but they do not shrink the range already set by the first layer.

\begin{figure}[hbt!]
\centering
\includegraphics[width=7.5in]{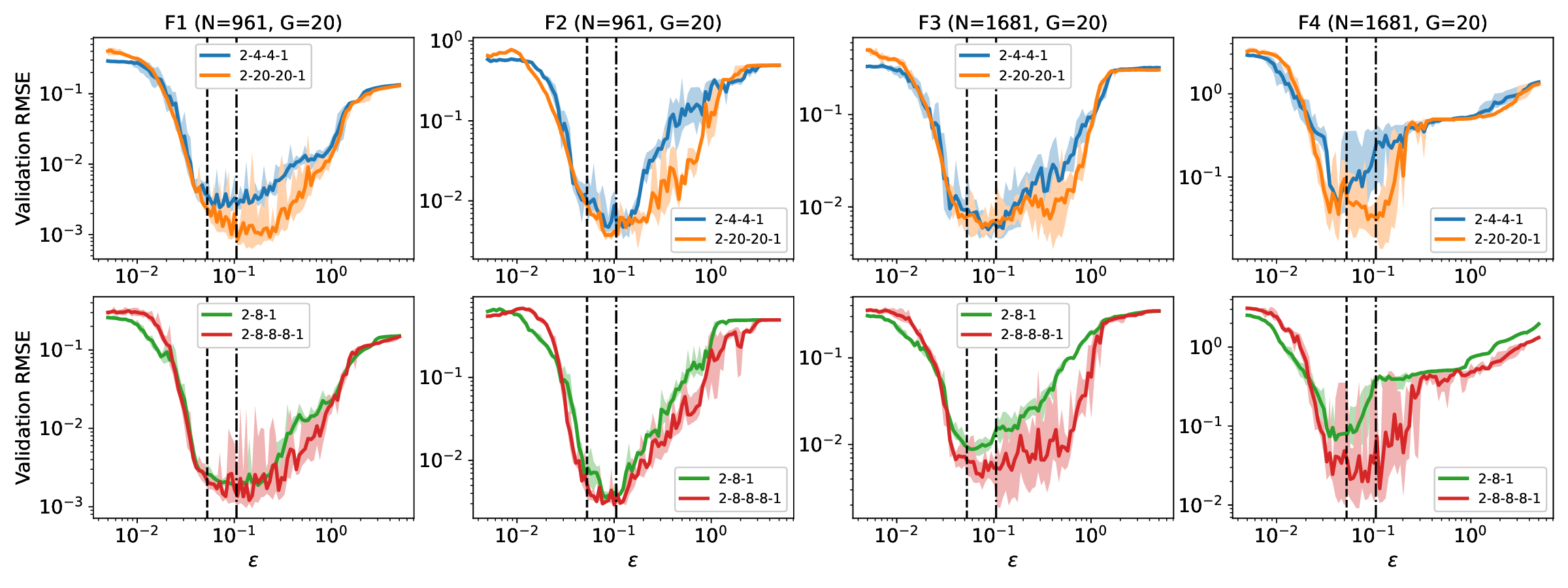}
\caption{
Validation RMSE versus the Gaussian scale \(\epsilon\) for different Gaussian KAN architectures at fixed grid size \(G=20\). The two vertical markers indicate \(\epsilon=1/(G-1)\) and \(\epsilon=2/(G-1)\). 
}
\label{fig:error_vs_eps_arch}
\end{figure}

\subsection{Effect of Input Dimension}

We finally investigate how the Gaussian scale behaves as the input dimension increases. For each dimension \(d\), we consider the smooth target
\begin{equation}
f_d(x)
=
\exp\!\left(
\frac{1}{d}
\sum_{i=1}^{d}
\left[
\sin(\pi x_i)+\frac{1}{2}x_i^2
\right]
\right),
\qquad
x\in[0,1]^d,
\label{eq:dimension_test_function}
\end{equation}
which remains bounded and of comparable magnitude across dimensions. The experiments in Figure~\ref{fig:HD_rmse_vs_eps} use the settings listed in Table~\ref{tab:dimension_experimental_settings}.

\begin{table}[hbt!]
\centering
\caption{Experimental settings used in the input-dimension study. The last column lists the hidden-layer widths, so for example $[12]$ denotes one hidden layer with 12 neurons, while $[12,12,12]$ denotes three hidden layers with 12 neurons each.}
\label{tab:dimension_experimental_settings}
\begin{tabular}{l l l}
\hline
Dimension $d$ & Collocation points $N$ & Hidden-layer widths \\
\hline
1 & 30    & $[12]$ \\
2 & 1000  & $[12,12]$ \\
3 & 30000 & $[12,12,12]$ \\
4 & 50000 & $[12,12,12,12]$ \\
\hline
\end{tabular}
\end{table}

Figure~\ref{fig:HD_rmse_vs_eps} shows that the validation RMSE keeps the same U-shaped dependence on \(\epsilon\) from 1D to 4D. Very small \(\epsilon\) leads to overly localized features, while very large \(\epsilon\) leads to the flat and numerically unfavorable regime. Most importantly, the low-error region remains close to the same two reference lines
so increasing the ambient dimension does not introduce a new preferred scale.

This is consistent with the first-layer matrix construction in \eqref{eq:full_first_layer_matrix}. The feature matrix is built by concatenating coordinate-wise Gaussian blocks,
\begin{equation}
\Phi=
\begin{bmatrix}
\Phi^{(1)} & \cdots & \Phi^{(d)}
\end{bmatrix},
\label{eq:dimension_feature_block_structure}
\end{equation}
so increasing \(d\) adds more one-dimensional blocks but does not change the spacing mechanism inside each block. The same shared scale \(\epsilon\) still controls the overlap between neighboring centers, and therefore the same conditioning-based interval remains relevant.

The minimum error does increase with dimension, as expected, because the approximation problem becomes harder. However, the shape of the error curve and the location of the useful \(\epsilon\)-window remain essentially unchanged. Thus, in the present separable Gaussian KAN, the selection of \(\epsilon\) is still controlled primarily by the first-layer Gaussian geometry rather than by the ambient dimension itself.

\begin{figure}[hbt!]
\centering
\includegraphics[width=7.5in]{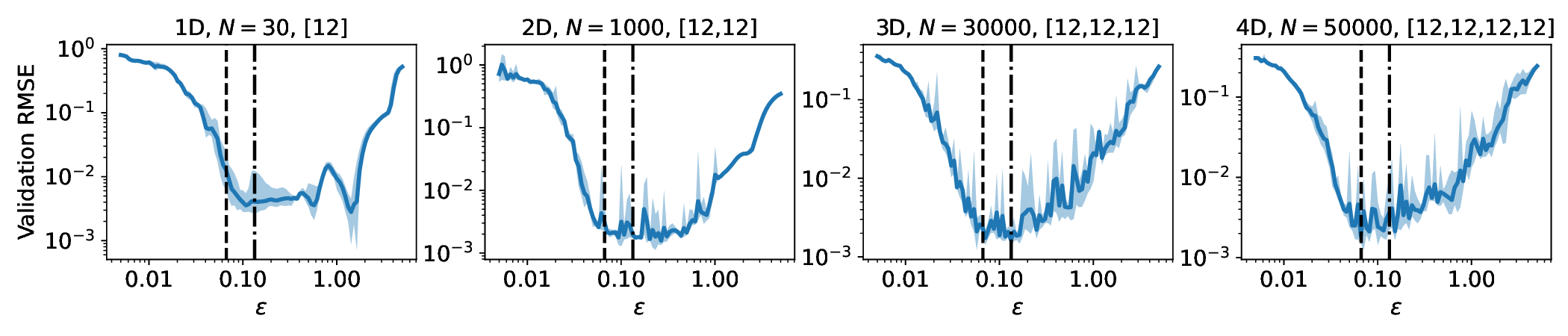}
\caption{
Validation RMSE versus the Gaussian scale \(\epsilon\) for input dimensions \(d=1,2,3,4\). 
The settings are those listed in Table~\ref{tab:dimension_experimental_settings}. 
The two vertical markers indicate the reference scales \(\epsilon=1/(G-1)\) and \(\epsilon=2/(G-1)\). 
Across dimensions, the same U-shaped dependence on \(\epsilon\) is observed, and the low-error region remains close to the same conditioning-based interval.
}
\label{fig:HD_rmse_vs_eps}
\end{figure}

\subsection{Gaussian KAN with Shared Centerwise Variable Scales}

The fixed-scale Gaussian feature map in \eqref{eq:gaussian_feature_map} can be generalized by assigning a different scale to each center while still sharing the same scale vector across all neurons and all layers. Let
\begin{equation}
\mathcal{E}=\{\epsilon_1,\dots,\epsilon_G\},
\qquad
\epsilon_g>0,
\quad g=1,\dots,G,
\label{eq:shared_variable_epsilon_vector}
\end{equation}
be the vector of centerwise Gaussian scales associated with the shared centers \(\mathcal{C}=\{c_1,\dots,c_G\}\subset[0,1]\). Then the Gaussian feature map becomes
\begin{equation}
\varphi_{\mathcal{E}}(t)
=
\begin{bmatrix}
\exp\!\left(-\dfrac{(t-c_1)^2}{\epsilon_1^2}\right)\\
\vdots\\
\exp\!\left(-\dfrac{(t-c_G)^2}{\epsilon_G^2}\right)
\end{bmatrix}
\in\mathbb{R}^G,
\qquad t\in\mathbb{R}.
\label{eq:variable_gaussian_feature_map}
\end{equation}

Thus, each center \(c_g\) has its own width \(\epsilon_g\), while the same vector \(\mathcal{E}\) is reused throughout the whole network. In the experiments below, the variable scales are sampled uniformly on the linear scale from the interval
\begin{equation}
\epsilon_g\in[\epsilon_{\min},\epsilon_{\max}],
\qquad
\epsilon_{\max}=\frac{2}{G-1},
\label{eq:variable_epsilon_safe_interval}
\end{equation}
and then assigned randomly to the shared centers \(\{c_g\}_{g=1}^G\). Hence, this model keeps the same Gaussian KAN architecture, but introduces controlled nonuniformity in the basis widths across centers. Although the centers themselves remain uniformly distributed, allowing the widths to vary makes the \emph{effective} coverage of the basis nonuniform across the domain. In this sense, the model gains some of the flexibility usually associated with nonuniform or data-adapted center placement, while still retaining the simplicity and interpretability of a fixed uniform center grid. Figure~\ref{compare_val_rmse_vs_epoch} compares this variable-scale model with the fixed-scale choice \(\epsilon=1/(G-1)\).

As Figure~\ref{compare_val_rmse_vs_epoch} shows, the shared centerwise variable-scale model can be more accurate than the single-\(\epsilon\) version. In the present experiment, the scales are assigned randomly within the admissible range, so this result should be interpreted only as an initial demonstration rather than an optimized construction. The use of variable shape parameters has a long history in classical RBF methods, where nonuniform scales often improve accuracy over constant-scale formulations~\cite{Ling06,Ling20,Fasshauer07}, and the present results indicate that this advantage can also carry over to the neural-network setting. More generally, one may vary not only the interval but also the sampling law for the scales, for example by using normal, log-normal, log-uniform, gamma, or chi-square distributions~\cite{Ling06}. Thus, even with uniformly distributed centers, variable-scale Gaussian KANs can induce a more adaptive and spatially nonuniform representation, making them a flexible and promising extension of the fixed-scale model.

\begin{figure}[hbt!]
\centering
\includegraphics[width=7in]{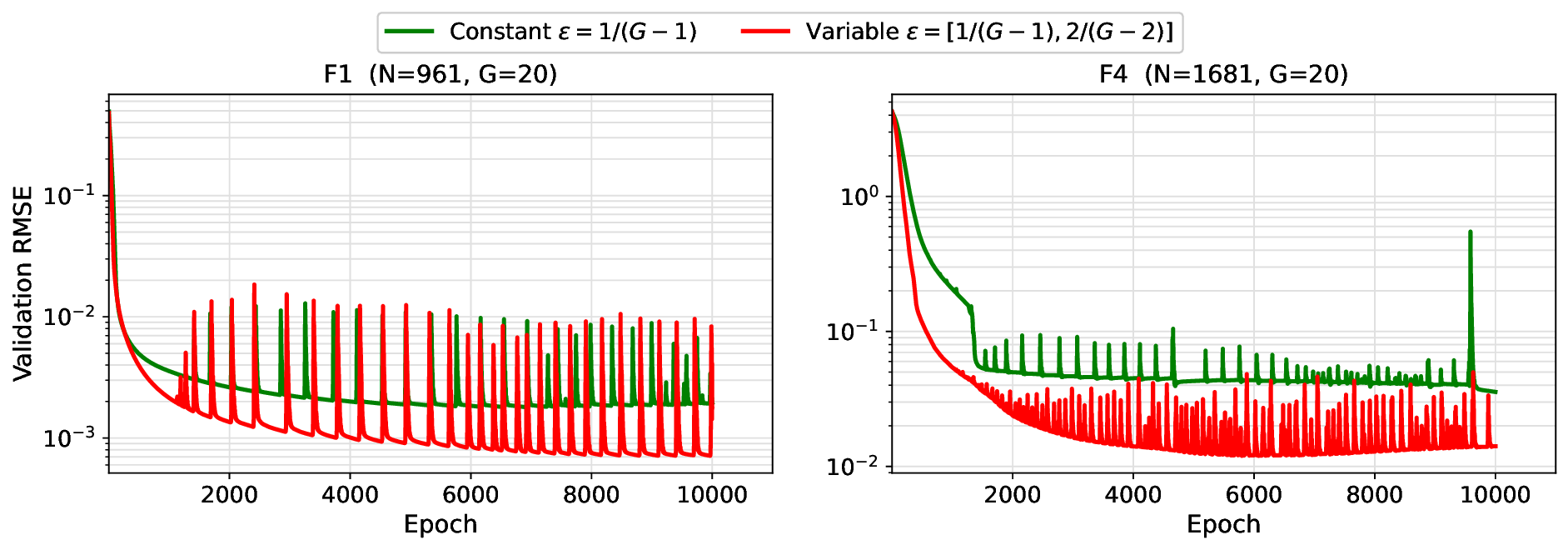}
\caption{Validation RMSE versus epoch for the fixed-scale Gaussian KAN and the Gaussian KAN with shared variable scales. Left: \(F1\) with \(N=961\). Right: \(F4\) with \(N=1681\). The fixed model uses the common scale \(\epsilon=1/(G-1)\), while the variable model assigns one scale \(\epsilon_g\) to each center \(c_g\), with the same scale vector shared across all neurons and all layers.}
\label{compare_val_rmse_vs_epoch}
\end{figure}

\subsection{Training-MSE-Based Scale Search Within the Admissible Interval}

The first-layer diagnostics do not select a single value of \(\epsilon\); they reduce the search to a numerically admissible interval. 
This is the main practical role of the range: instead of sweeping over all candidate scales, one only needs to search inside the conditioning-controlled region. 
Figure~\ref{fig:eps_train_val_proxy} shows that, within this reduced interval, the training MSE provides a useful criterion for selecting a working scale.

For the four representative problems \(F1\), \(F2\), \(F3\), and \(F4\), the low-training-error region at \(500\) epochs is located in essentially the same part of the \(\epsilon\)-axis as the low-validation-RMSE region at \(10000\) epochs. 
The agreement is clearest for the smooth and oscillatory cases \(F1\) and \(F2\), and remains informative for the more difficult cases \(F3\) and \(F4\), where the curves are rougher but still identify the same favorable range.

The practical implications are straightforward:
\begin{itemize}
\item The admissible interval sharply reduces the \(\epsilon\)-search space.
\item In these regression experiments, early training MSE provides a useful proxy for localizing a good scale within the admissible interval.
\item Hence one can perform a short sweep over \(\epsilon\), choose a candidate from the training-MSE curve, and run full training only for that value.
\item This is especially useful when the exact solution is unavailable, so validation against the true target cannot be used for scale selection.
\end{itemize}

Therefore, the proposed interval is useful not only for stability control, but also for reducing the cost of scale selection. 
It turns the search for \(\epsilon\) into a smaller optimization problem over a numerically justified range, and the early training MSE is usually sufficient to identify the relevant part of that range.

\begin{figure}[hbt!]
\centering
\includegraphics[width=7in]{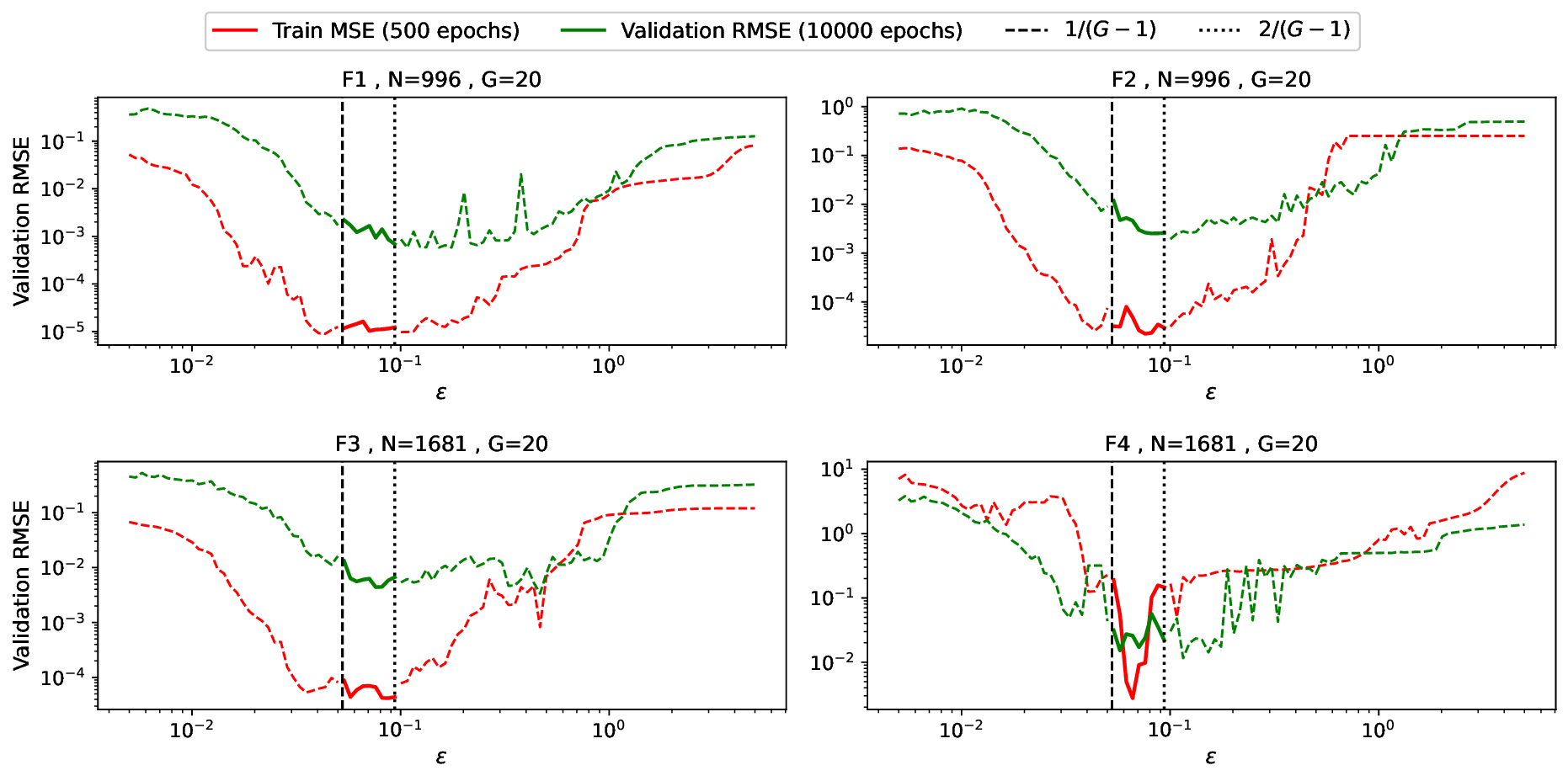}
\caption{
Training-MSE-based scale selection inside the admissible interval. 
Each panel shows the training MSE at \(500\) epochs and the validation RMSE at \(10000\) epochs as functions of \(\epsilon\). 
The vertical markers indicate the lower reference scale and the conditioning-based upper boundary of the search interval. 
Across all four examples, the early training-MSE curve already localizes essentially the same favorable \(\epsilon\)-region as the final validation RMSE, so the admissible interval can be searched efficiently using only a short preliminary sweep.
}
\label{fig:eps_train_val_proxy}
\end{figure}

\subsection{Helmholtz Problem: Physics-Informed Validation}

To test whether the conditioning-based scale interval also remains useful in a physics-informed setting, we consider the two-dimensional Helmholtz problem
\begin{equation}
-\Delta u - \lambda u = f
\qquad \text{in } \Omega=[0,1]^2,
\label{eq:helmholtz_pde}
\end{equation}
with homogeneous Dirichlet boundary condition
\begin{equation}
u=0
\qquad \text{on } \partial\Omega.
\label{eq:helmholtz_bc}
\end{equation}
We choose the exact solution
\begin{equation}
u(x,y)
=
\sin(a_1\pi x)\sin(a_2\pi y),
\qquad
(x,y)\in[0,1]^2,
\label{eq:helmholtz_exact_solution}
\end{equation}
which vanishes on the boundary and therefore satisfies \eqref{eq:helmholtz_bc}. 
Substituting \eqref{eq:helmholtz_exact_solution} into \eqref{eq:helmholtz_pde} gives the forcing term
\begin{equation}
f(x,y)
=
\Bigl((a_1^2+a_2^2)\pi^2-\lambda\Bigr)
\sin(a_1\pi x)\sin(a_2\pi y).
\label{eq:helmholtz_forcing}
\end{equation}

In the experiments below, the PDE residual is weighted by
\begin{equation}
w_{\mathrm{pde}}=1,
\label{eq:helmholtz_pde_weight}
\end{equation}
and the boundary loss is weighted by
\begin{equation}
w_{\mathrm{bc}}=100.
\label{eq:helmholtz_bc_weight}
\end{equation}
The collocation size and Gaussian grid are fixed at
\begin{equation}
(N_{BC},N_{PDE})=(800,2000),
\qquad
G=20,
\label{eq:helmholtz_NG}
\end{equation}
so the conditioning-based reference interval from \eqref{eq:practical_epsilon_interval} becomes
\begin{equation}
\epsilon\in
\left[
\frac{1}{19},
\frac{2}{19}
\right]
\approx
[0.0526,\,0.1053].
\label{eq:helmholtz_reference_interval}
\end{equation}
We study three values of the Helmholtz parameter,
\begin{equation}
\lambda\in\{0,10,100\},
\label{eq:helmholtz_lambda_values}
\end{equation}
for two exact solutions: the lower-frequency case
\begin{equation}
(a_1,a_2)=(1,2),
\label{eq:helmholtz_case_11}
\end{equation}
and the anisotropic higher-frequency case
\begin{equation}
(a_1,a_2)=(1,4).
\label{eq:helmholtz_case_14}
\end{equation}

Figure~\ref{fig:helmholtz_eps} shows the validation RMSE as a function of the Gaussian scale \(\epsilon\), with the two vertical markers corresponding to the endpoints of \eqref{eq:helmholtz_reference_interval}. 
In both test cases, the same qualitative behavior seen in the regression experiments is recovered: the error is large for very small \(\epsilon\), decreases sharply into a low-error regime, and then increases again as \(\epsilon\) becomes too large. 
Thus, the Helmholtz problem exhibits the same basic trade-off between insufficient overlap at small scales and degradation of the Gaussian representation at large scales.

Most importantly, the best-performing region remains within or close to the interval \eqref{eq:helmholtz_reference_interval} for all three values of \(\lambda\). 
For \(\lambda=0\) and \(\lambda=10\), the minimum error is attained clearly inside this interval in both problem settings. 
For the more challenging case \(\lambda=100\), the error level is higher and the low-error region becomes broader and less stable, but it is still centered near the same conditioning-based range. 
This indicates that the first-layer scale rule derived earlier remains informative even after the loss is replaced by a PDE residual and boundary penalty.

The two exact solutions also illustrate a useful distinction. 
The case \((a_1,a_2)=(1,2)\) is smoother and easier, so the minimum RMSE is lower and the low-error region is broader. 
The case \((a_1,a_2)=(1,4)\) is more oscillatory in the \(y\)-direction, and therefore more demanding, which raises the minimum error and makes the curves more sensitive, especially for larger \(\lambda\). 
Nevertheless, across both exact solutions and all tested values of \(\lambda\), the preferred \(\epsilon\)-window remains within or close to the same conditioning-based interval and stays tied to the same geometric scale set by the shared Gaussian centers. This shows that larger \(\lambda\) and more oscillatory solutions may increase the error and make the curves less stable, but they do not substantially change the preferred \(\epsilon\)-window.

This example is intentionally limited to a single PDE family with a tunable parameter \(\lambda\). 
Its role is not to provide an exhaustive PDE study, but to show that the interval suggested by the first-layer conditioning analysis continues to work in a representative physics-informed problem. 
A broader investigation across other PDEs, boundary conditions, and training regimes can be carried out in future work.

\begin{figure}[hbt!]
\centering
\subfigure[]{\label{fig:helmholtz_eps_12}
\includegraphics[width=3in]{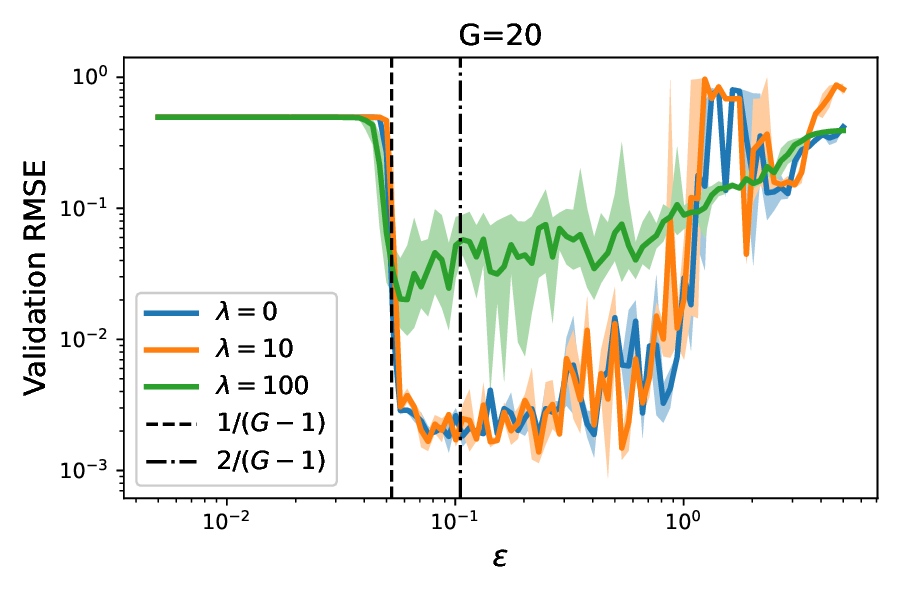}}
\subfigure[]{\label{fig:helmholtz_eps_14}
\includegraphics[width=3in]{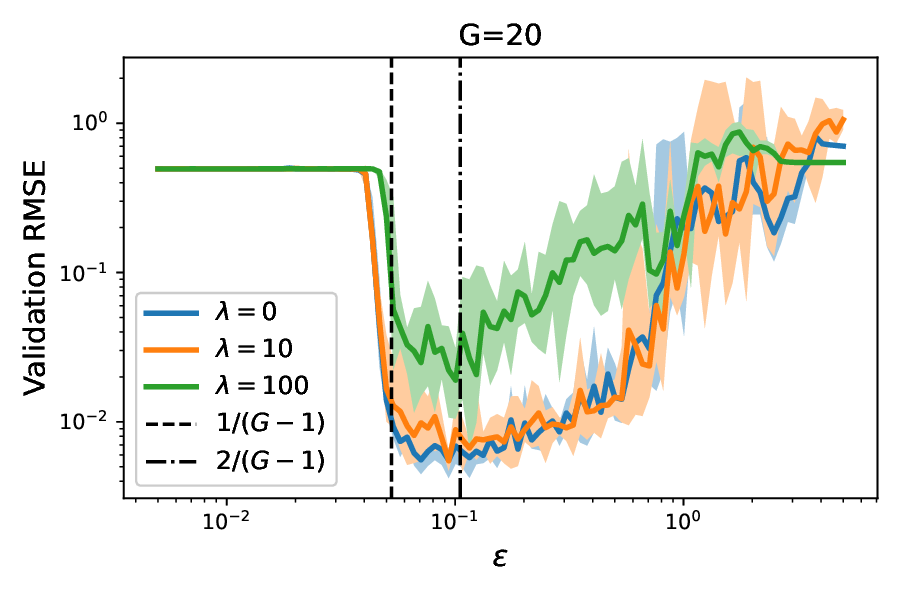}}
\caption{
Validation RMSE versus the Gaussian scale \(\epsilon\) for the Helmholtz problem \eqref{eq:helmholtz_pde}--\eqref{eq:helmholtz_bc} with \((N_{BC},N_{PDE})=(800,2000)\), \(G=20\), \(w_{\mathrm{pde}}=1\), and \(w_{\mathrm{bc}}=100\). 
The two vertical markers indicate \(\epsilon=1/(G-1)\) and \(\epsilon=2/(G-1)\), i.e. the endpoints of the conditioning-based interval \eqref{eq:helmholtz_reference_interval}. 
\textbf{(a)} Exact solution \eqref{eq:helmholtz_exact_solution} with \((a_1,a_2)=(1,2)\). 
\textbf{(b)} Exact solution \eqref{eq:helmholtz_exact_solution} with \((a_1,a_2)=(1,4)\). 
In both cases, and for all tested values \(\lambda\in\{0,10,100\}\), the best-performing region lies within or close to the same conditioning-based range.
}
\label{fig:helmholtz_eps}
\end{figure}

\subsection{Black--Scholes Digital Option Problem}

As a second physics-informed test, we consider the Black--Scholes equation for a digital cash-or-nothing call option,
\begin{equation}
u_t+\frac{1}{2}\sigma^2 S^2 u_{SS}+rS u_S-r u=0,
\qquad (S,t)\in[0,S_{\max}]\times[0,T],
\label{eq:bs_pde}
\end{equation}
where \(S\) denotes the asset price and \(t\) denotes time. We use
\begin{equation}
S_{\max}=1,\qquad T=1,\qquad r=0.05,\qquad \sigma=0.2,\qquad K=0.5.
\label{eq:bs_parameters}
\end{equation}
The terminal condition is the discontinuous digital payoff
\begin{equation}
u(S,T)=
\begin{cases}
1, & S>K,\\
0, & S\le K,
\end{cases}
\label{eq:bs_terminal_condition}
\end{equation}
together with the boundary conditions
\begin{equation}
u(0,t)=0,
\qquad
u(S_{\max},t)=u_{\mathrm{exact}}(S_{\max},t).
\label{eq:bs_boundary_conditions}
\end{equation}
For \(t<T\), the exact solution is given by
\begin{equation}
u_{\mathrm{exact}}(S,t)
=
\exp\bigl(-r(T-t)\bigr)\Phi(d_2),
\label{eq:bs_exact_solution}
\end{equation}
where \(\Phi(\cdot)\) denotes the standard normal cumulative distribution function and
\begin{equation}
d_2=
\frac{
\log(S/K)+\left(r-\frac{1}{2}\sigma^2\right)(T-t)
}{
\sigma\sqrt{T-t}
}.
\label{eq:bs_d2}
\end{equation}
The expression \eqref{eq:bs_d2} is singular at \(t=T\), and the terminal value is therefore understood through the payoff condition \eqref{eq:bs_terminal_condition}. This problem is more challenging than the smooth Helmholtz example because the terminal data have a jump discontinuity at \(S=K\). Therefore, it provides a useful test of whether the proposed scale interval remains informative for a PDE problem with a sharp payoff layer.

In the experiment, we use \(N_{\mathrm{PDE}}=3000\) interior collocation points and \(N_{\mathrm{BC}}=500\) boundary/terminal points, with boundary weight \(w_{\mathrm{bc}}=1\). We compare three Gaussian grids, \(G=16\), \(G=20\), and \(G=24\). 
Figure~\ref{fig:bs_eps} shows the validation RMSE as a function of the Gaussian scale \(\epsilon\). For all  \(G=16\), \(G=20\), and \(G=24\), the error is large for very small \(\epsilon\), decreases sharply once the basis functions begin to overlap sufficiently, and then increases again for overly large \(\epsilon\). Despite the nonsmooth terminal payoff, the lowest-error region remains close to the interval predicted by the first-layer scale rule. Although larger values of \(\epsilon\) can still achieve comparable accuracy in some cases, the overall trend shows increasing error as \(\epsilon\) becomes too large.

\begin{figure}[hbt!]
\centering
\includegraphics[width=7in]{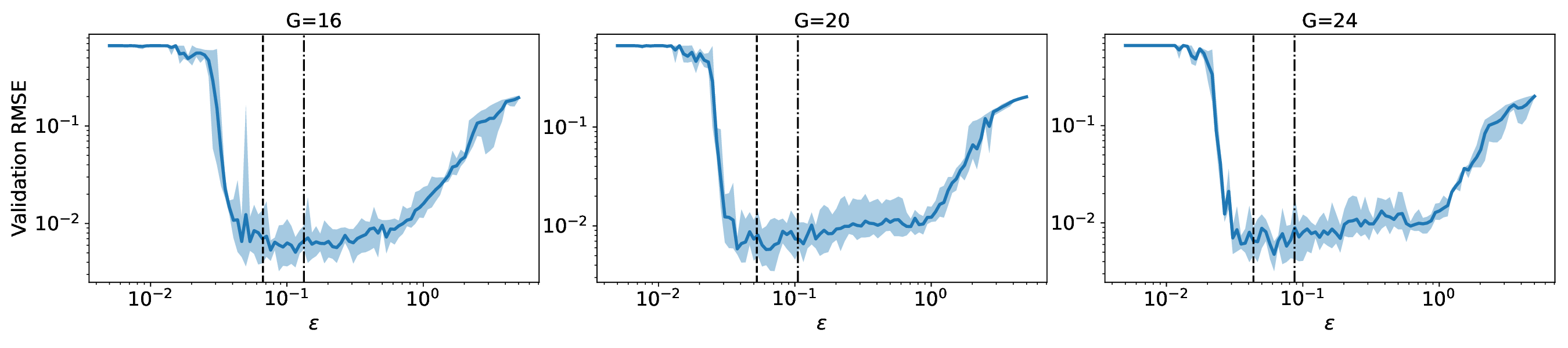}
\caption{
Validation RMSE versus the Gaussian scale \(\epsilon\) for the Black--Scholes digital option problem \eqref{eq:bs_pde}--\eqref{eq:bs_boundary_conditions}.
The three panels correspond to \(G=16\), \(G=20\), and \(G=24\). 
The vertical markers indicate \(\epsilon=1/(G-1)\) and \(\epsilon=2/(G-1)\). 
In all three cases, after \(\epsilon=2/(G-1)\), the error begins to increase as \(\epsilon\) becomes larger, and much larger values of \(\epsilon\) lead to substantially worse accuracy.
}
\label{fig:bs_eps}
\end{figure}

\subsection{Practical Scale Interval for Mat\'ern Bases}
\label{subsec:matern_scale_interval}

We next check whether the scale-selection viewpoint developed for Gaussian KANs extends to another RBF family. 
For this purpose, we consider a Mat\'ern KAN in which each edge function is expanded using a finite Mat\'ern basis with \(G\) centers \(c_g\in[0,1]\), \(g=1,\ldots,G\). 
For a scalar input \(t\in[0,1]\), define
\begin{equation}
d_g(t)=|t-c_g|,
\qquad
r_g(t)=\frac{\sqrt{2\nu}\,d_g(t)}{\epsilon},
\label{eq:matern_scaled_distance}
\end{equation}
where \(\nu\) is the Mat\'ern smoothness parameter and \(\epsilon>0\) is the scale parameter. 
The Mat\'ern feature map is written as
\begin{equation}
\varphi_{\mathrm M,\nu}(t)
=
\begin{bmatrix}
P_\nu(r_1(t))e^{-r_1(t)}\\
\vdots\\
P_\nu(r_G(t))e^{-r_G(t)}
\end{bmatrix},
\label{eq:matern_feature_map}
\end{equation}
where, for the closed-form cases used here,
\begin{equation}
P_1(r)=1,\qquad
P_3(r)=1+r,\qquad
P_5(r)=1+r+\frac{r^2}{3}.
\label{eq:matern_polynomial_cases}
\end{equation}
Thus, the Mat\'ern KAN layer has the same finite-feature form as the Gaussian KAN layer, but with the Gaussian feature map replaced by \(\varphi_{\mathrm M,\nu}\).

In the experiments below, we use the Mat\'ern-5 basis, corresponding to \(\nu=5\),
\begin{equation}
\phi_{\mathrm{M5}}(d;\epsilon)
=
\left(1+r+\frac{r^2}{3}\right)e^{-r},
\qquad
r=\frac{\sqrt{10}\,d}{\epsilon},
\label{eq:matern5_basis_profile_main}
\end{equation}
where \(d\ge 0\) denotes the distance from a center. 
Since the centers are uniformly distributed on \([0,1]\), the adjacent-center spacing is
\begin{equation}
h=\frac{1}{G-1}.
\label{eq:matern_spacing_h}
\end{equation}
Using the same adjacent-overlap level \(e^{-1}\) as in the Gaussian case gives
\begin{equation}
\epsilon_{\mathrm{low}}^{\mathrm{M5}}
\approx
1.0887\,h
=
\frac{1.0887}{G-1}.
\label{eq:matern5_lower_main}
\end{equation}
The derivation of this constant is given in Appendix~\ref{app:matern_lower_bound}.

The upper endpoint is again determined by first-layer conditioning. 
Let
\begin{equation}
\Phi^{\mathrm M}_{\epsilon,5}\in\mathbb{R}^{N\times dG}
\label{eq:matern_first_layer_main}
\end{equation}
denote the first-layer Mat\'ern-5 feature matrix generated by \eqref{eq:matern5_basis_profile_main} before multiplication by the learned coefficients, where \(N\) is the number of input samples and \(d\) is the input dimension. 
We define
\begin{equation}
\epsilon_{\kappa}^{\mathrm M}
=
\arg\min_{\epsilon>0}
\left|
\log \kappa\!\left(\Phi^{\mathrm M}_{\epsilon,5}\right)
-
\log\!\left(3\times 10^3\right)
\right|,
\label{eq:matern_eps_kappa_main}
\end{equation}
where \(\kappa(\cdot)\) denotes the matrix condition number. 
The empirical safe region for the Mat\'ern-5 KAN is therefore
\begin{equation}
\boxed{
\frac{1.0887}{G-1}
\;\lesssim\;
\epsilon
\;\lesssim\;
\epsilon_{\kappa}^{\mathrm M}
}
\label{eq:matern5_safe_interval_main}
\end{equation}
where the lower endpoint is determined by adjacent-center overlap and the upper endpoint is determined by the first-layer conditioning threshold.

Figure~\ref{fig:Matern_Cond_safe} verifies this interval empirically. 
Across the four test cases, the useful accuracy region lies between the lower overlap scale and the conditioning-based marker \(\epsilon_{\kappa}^{\mathrm M}\). 
As in the Gaussian case in Figure~\ref{fig:Cond_safe}, the best RMSE occurs before the first-layer feature matrix becomes strongly ill-conditioned, and the rank-loss marker is close to the conditioning marker. 
Compared with the Gaussian basis, the Mat\'ern-5 basis remains well conditioned over a wider range of \(\epsilon\), so the onset of rapid condition-number growth occurs at larger scale values. 
This is consistent with the better conditioning behavior of Mat\'ern-type bases observed in classical meshfree approximation~\cite{Wendland05,Fasshauer07}.

\begin{figure}[hbt!]
\centering
\subfigure[]{\label{eps_diag_matern5_F1_N961_G10}
\includegraphics[width=3.3in]{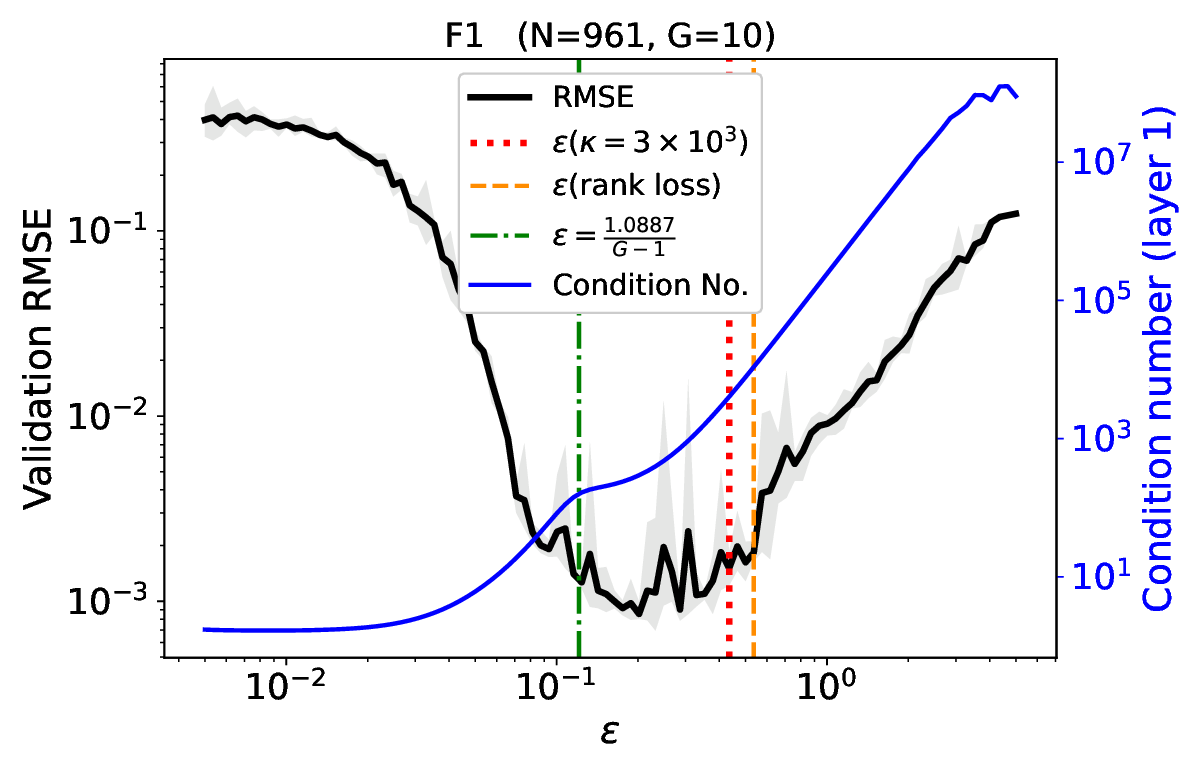}}
\subfigure[]{\label{eps_diag_matern5_F11_N961_G14}
\includegraphics[width=3.3in]{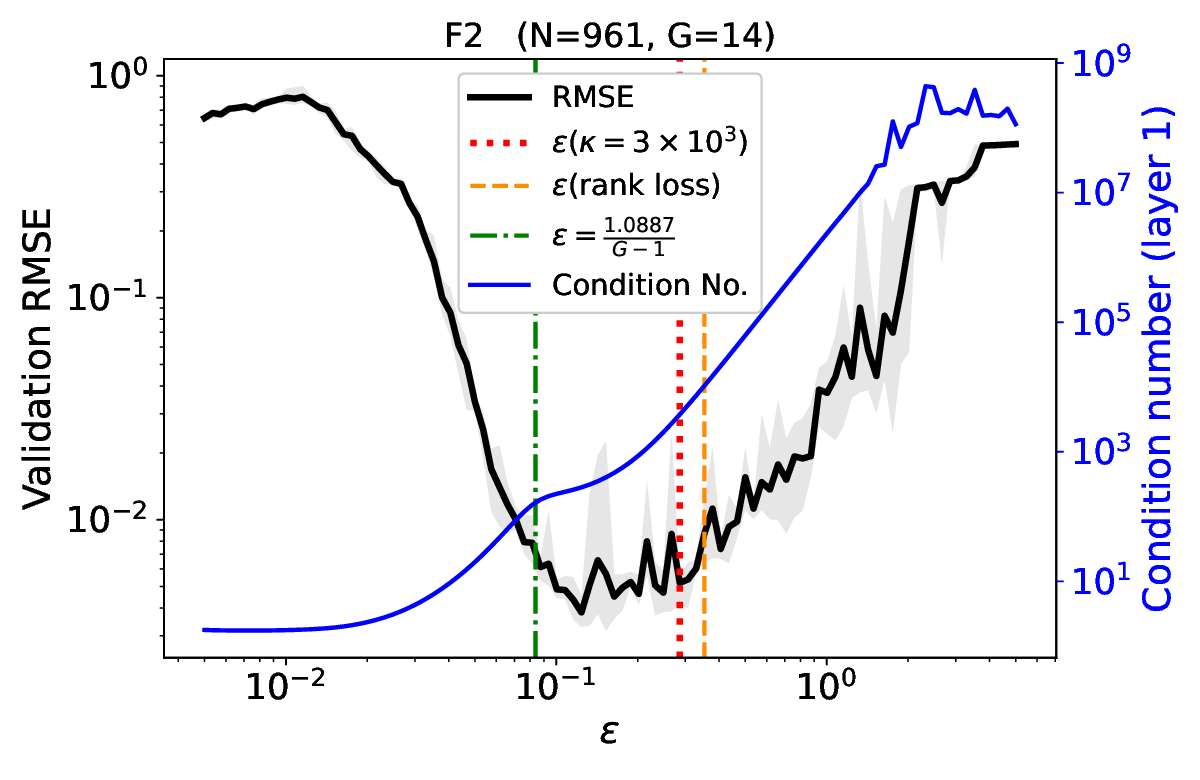}}
\subfigure[]{\label{eps_diag_matern5_F20_N1681_G16}
\includegraphics[width=3.3in]{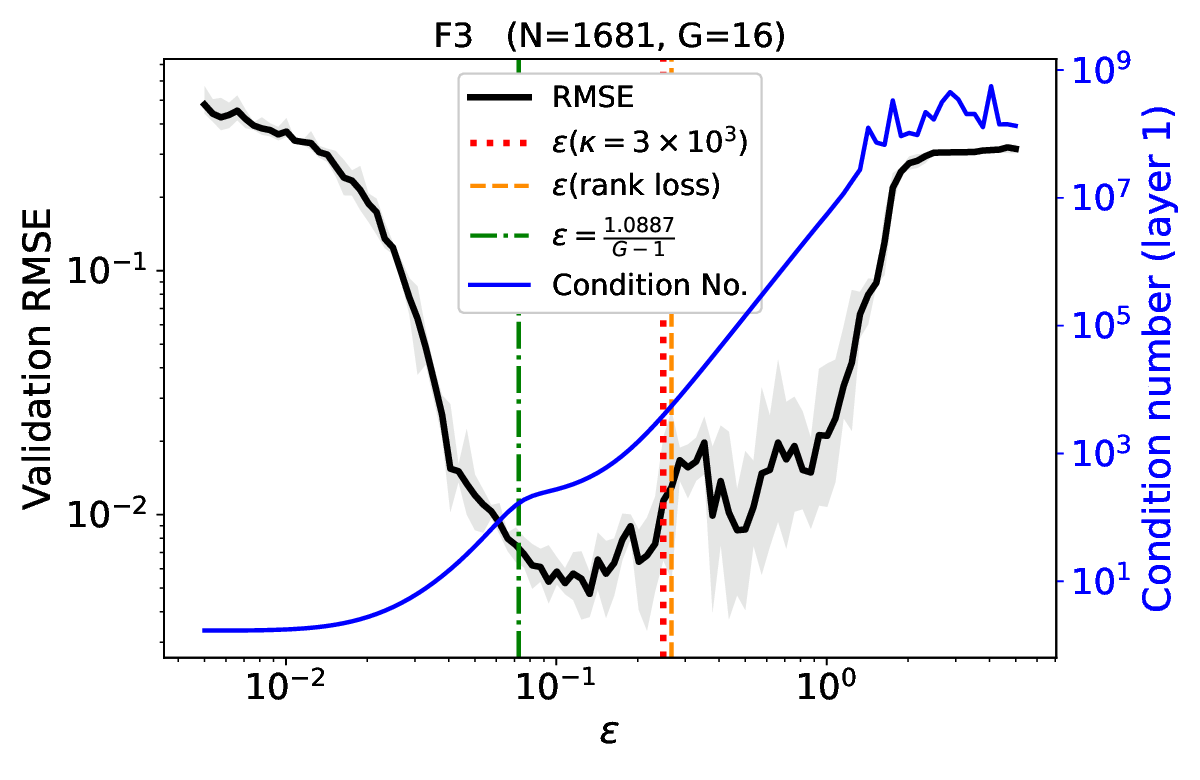}}
\subfigure[]{\label{eps_diag_matern5_F24_N1681_G20}
\includegraphics[width=3.3in]{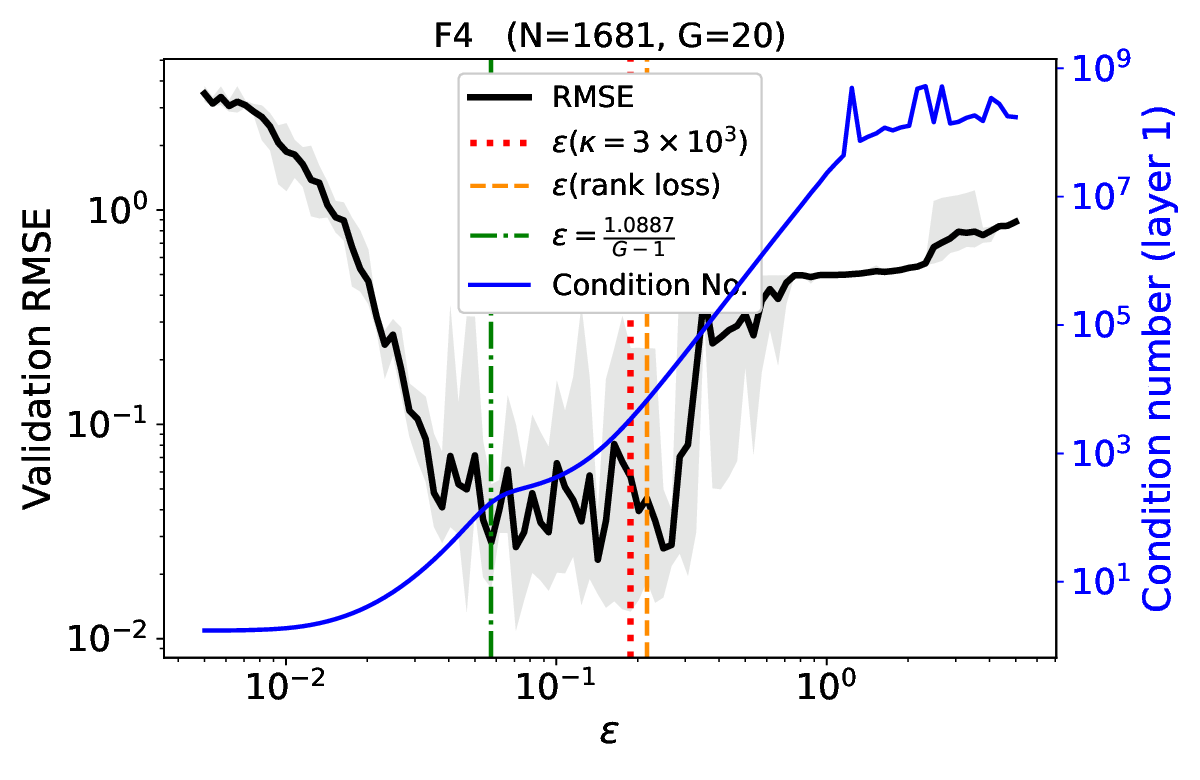}}
\caption{
Empirical localization of the Mat\'ern-5 scale interval. Each panel shows the validation RMSE, the first-layer condition number \(\kappa(\Phi^{\mathrm M}_{\epsilon,5})\), the lower overlap scale \(\epsilon=1.0887/(G-1)\), the conditioning marker \(\epsilon_{\kappa}^{\mathrm M}\), and the rank-loss marker. The best accuracy occurs between the lower overlap scale and the conditioning boundary.
}
\label{fig:Matern_Cond_safe} 
\end{figure}

\section{Conclusion}\label{sec:conclusion}

This work studies Gaussian KANs from a kernel-feature and conditioning viewpoint, with particular emphasis on the role of the Gaussian scale parameter \(\epsilon\). The results show that an appropriate choice of \(\epsilon\) is essential for both numerical stability and approximation accuracy. Under a matched number of trainable coefficients per edge, a properly scaled Gaussian KAN achieves accuracy comparable to, and in several cases better than, the corresponding Chebyshev KAN. This comparison is intended to demonstrate the competitiveness of Gaussian KANs when the scale is selected properly, rather than to provide a definitive benchmark between basis families.

The analysis identifies the first layer as the critical layer for scale selection because it is constructed directly on the input domain, and any loss of feature distinguishability introduced there propagates through the remaining network. Accordingly, the main stability quantity is the first-layer feature matrix \(\Phi\). Based on its numerical rank and conditioning, we identify the practical interval
\[
\epsilon \in \left[\frac{1}{G-1},\,\frac{2}{G-1}\right],
\]
where \(G\) is the number of Gaussian centers. The numerical experiments show a clear U-shaped error dependence on \(\epsilon\), with the best-performing region lying within or close to this interval. The rule remains informative across different target functions, collocation densities, grid sizes, network architectures, input dimensions, and physics-informed Helmholtz and digital-option problems. A similar conditioning-based behavior is also observed for Mat\'ern KANs.

The layer-wise experiments further show that varying the first-layer scale produces behavior close to that of the fully shared-\(\epsilon\) model, while changing only the deeper-layer scales has a weaker effect. Moreover, using separate scales across layers does not provide a clear improvement in the best attainable accuracy. Therefore, using one well-selected scale in all layers offers a simple and effective design choice.

The experiments also clarify the effects of the number of collocation points \(N\) and centers \(G\). Increasing \(N\) generally improves the attainable accuracy when \(\epsilon\) is properly selected. Increasing \(G\) mainly shifts the useful scale region toward smaller values and broadens the admissible range, but does not necessarily produce a substantially lower minimum error. Thus, a relatively small \(G\) may provide nearly the same accuracy at a lower computational cost.

Beyond fixed-scale selection, the proposed interval supports variable-scale constructions, constrained optimization of \(\epsilon\), and efficient scale search. In particular, early-stage training MSE provides a useful proxy for identifying a suitable scale when the exact solution is unavailable. Overall, the Gaussian scale parameter should not be treated as an arbitrary hyperparameter, but as a central design quantity that can be selected systematically to improve stability, accuracy, and computational efficiency.

\section{Limitations and Future Work}\label{sec:limitations}

This study focuses primarily on the conditioning of the first-layer feature matrix. Deeper layers may also become ill-conditioned because the transformed coordinates can be clustered or poorly aligned with the fixed Gaussian centers. Solving this issue would require more complex strategies, such as adaptive center relocation, multilevel training, or updating the collocation points during optimization. However, improved conditioning in deeper layers does not necessarily guarantee better accuracy, and the present results show that the standard Gaussian KAN remains effective when \(\epsilon\) is selected within the proposed range.

Several directions remain open. The proposed scale-selection rule should be further examined in classification, operator learning, and additional physics-informed problems. Its extension to other RBF-based KANs is also of interest~\cite{Wendland05,Iske96,Muller09}. Another direction is to use the proposed interval as a constraint when training a learnable \(\epsilon\), thereby avoiding unrestricted and potentially unstable scale optimization. Variable-scale Gaussian KANs also deserve further investigation, including different distributions and sampling strategies for the scale parameters. Since variable shape parameters have proved useful in classical RBF approximation~\cite{Ling06,Ling20,Fasshauer07}, they may provide further improvements in Gaussian KANs. Finally, testing more PDEs and larger problem settings would help assess the broader applicability of the proposed conditioning-based rule.

\section{Acknowledgments}
This work was supported by the General Research Fund (GRF No.12301824) of the Hong Kong Research Grant Council and the Guangdong and Hong Kong Universities ``1+1+1'' Joint Research Collaboration Scheme (project no. 2025A0505000014).

\appendix

\section{Additional comparison between Gaussian-KAN and Chebyshev-KAN}
\label{app:gaussian_chebyshev_appendix}

In this appendix, we provide an additional diagnostic comparison between Gaussian-KAN and Chebyshev-KAN on three representative two-dimensional target functions. The purpose of this experiment is not to establish a general superiority of one architecture over the other, but rather to illustrate the critical role of the Gaussian scale parameter \(\epsilon\). In all cases, the number of sample/collocation points is fixed as \(N=2000\), and the network architecture consists of two hidden layers with 12 neurons in each layer. For Gaussian-KAN, the number of grid points is fixed as \(G=20\), while \(\epsilon\) is varied. For Chebyshev-KAN, several polynomial degrees are considered, and these degrees are denoted by \(D\) in the figure. The corresponding validation RMSE values are shown as horizontal reference lines. In all cases, the reported RMSE is the best validation RMSE obtained within 10000 training epochs. 

The three target functions used in this appendix are defined as
\begin{align}
F_{1A}(x,y)
&=
0.75\exp\!\left(-\frac{(9x-2)^2+(9y-2)^2}{4}\right)
+
0.75\exp\!\left(-\frac{(9x+1)^2}{49}-\frac{(9y+1)^2}{10}\right)
\nonumber\\
&\quad
+
0.5\exp\!\left(-\frac{(9x-7)^2+(9y-3)^2}{4}\right)
-
0.2\exp\!\left(-(9x-4)^2-(9y-7)^2\right),
\qquad (x,y)\in[0,1]^2,
\label{eq:F1A_appendix}
\\[1ex]
F_{2A}(x,y)
&=
\sin(4\pi x)\sin(4\pi y),
\qquad (x,y)\in[0,1]^2,
\label{eq:F2A_appendix}
\\[1ex]
F_{3A}(x,y)
&=
\cos\!\left(10(x^2+y)\right)
\sin\!\left(10(x+y^2)\right),
\qquad (x,y)\in[-1,1]^2.
\label{eq:F3A_appendix}
\end{align}

Figure~\ref{fig:Versus_Cheby_Appendix} shows the exact contours of these target functions in the first row and the corresponding validation RMSE results in the second row. A clear scale-dependent behavior is observed for Gaussian-KAN. When \(\epsilon\) is poorly selected, the Gaussian-KAN error can be relatively large; however, when \(\epsilon\) is chosen in a suitable range, Gaussian-KAN can achieve lower validation RMSE than several Chebyshev-KAN configurations. This behavior is visible for all three examples, where the Gaussian-KAN curve forms a pronounced valley as \(\epsilon\) varies.

This result supports the main motivation of the present study: the scale parameter in Gaussian-KAN is not a minor implementation detail, but a decisive factor controlling accuracy. In practice, an unsuitable choice of \(\epsilon\) may make Gaussian-KAN appear less competitive than polynomial-based KAN variants such as Chebyshev-KAN. Conversely, an appropriate scale can lead to substantially improved accuracy. Therefore, the comparison emphasizes the importance of systematic scale-parameter selection in Gaussian-KAN, rather than serving as a broad benchmark comparison between Gaussian and Chebyshev bases.

\begin{figure}[hbt!]
\centering
\includegraphics[width=\textwidth]{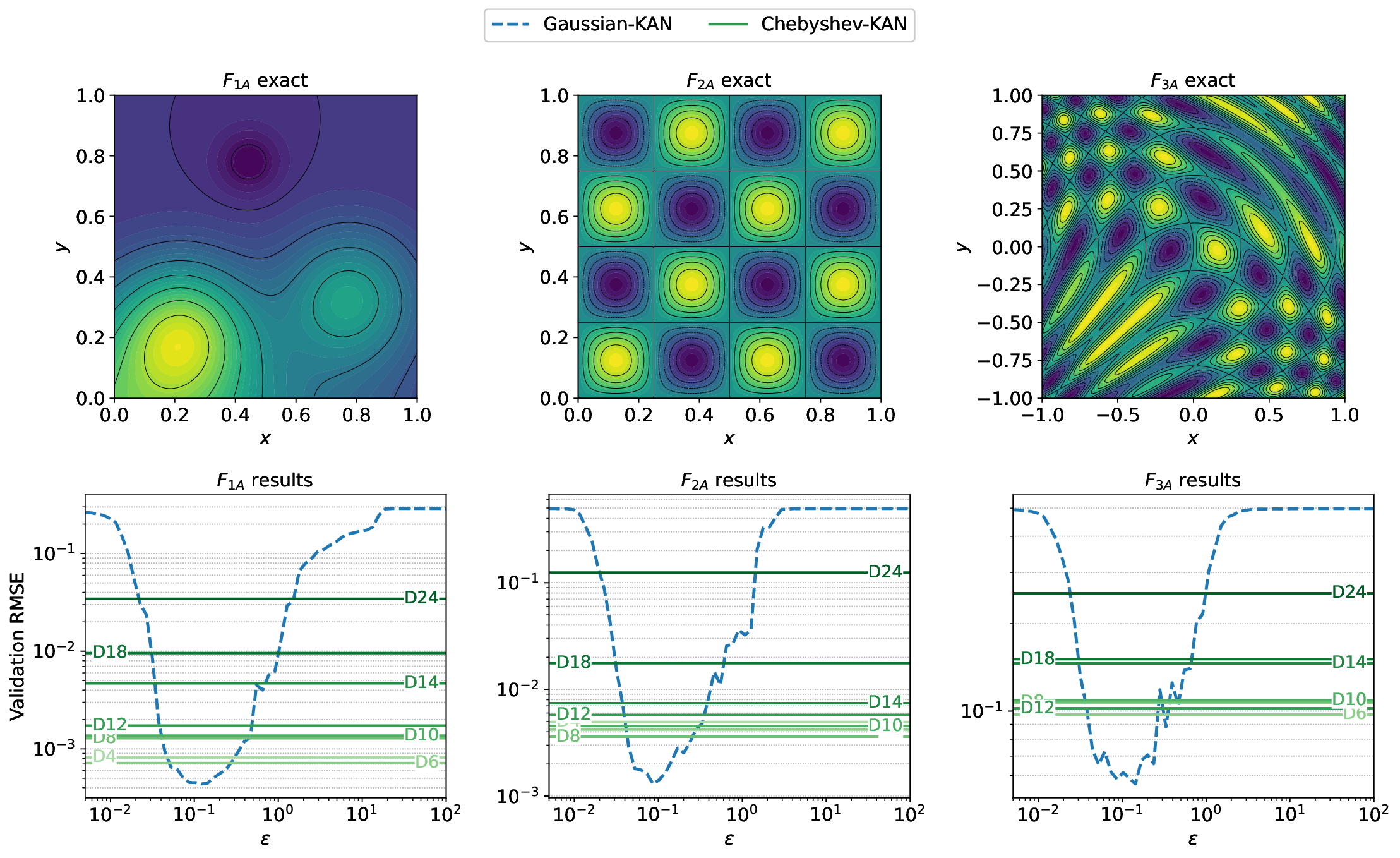}
    \caption{Additional appendix comparison between Gaussian-KAN and Chebyshev-KAN on the target functions \(F_{1A}\), \(F_{2A}\), and \(F_{3A}\). The first row shows the exact target contours, while the second row shows the best validation RMSE within 10000 epochs. Gaussian-KAN uses \(G=20\) and varying scale parameter \(\epsilon\); Chebyshev-KAN results are shown for different polynomial degrees.}
\label{fig:Versus_Cheby_Appendix}
\end{figure}

\section{Proof of Proposition~\ref{prop:first_layer_bottleneck}}
\label{app:first_layer_bottleneck_proof}

\begin{proof}
First, suppose that \(G(x)=G(x')\). 
Since the full network is written as
\[
f=T\circ G,
\]
we have
\[
f(x)=T(G(x)),
\qquad
f(x')=T(G(x')).
\]
Using \(G(x)=G(x')\), it follows immediately that
\[
f(x)=T(G(x))=T(G(x'))=f(x').
\]
This proves part (i).

Next, suppose that \(\Phi(x)=\Phi(x')\). 
The first-layer kernel is defined by the inner product of first-layer feature vectors,
\[
K_0(x,z)=\Phi(x)^\top\Phi(z).
\]
Therefore, for any \(z\in[0,1]^d\),
\[
K_0(x,z)
=
\Phi(x)^\top\Phi(z)
=
\Phi(x')^\top\Phi(z)
=
K_0(x',z).
\]
Thus, \(x\) and \(x'\) induce the same first-layer kernel section. 
Moreover, since \(\Phi(x)=\Phi(x')\), for any first-layer coefficient matrix \(W^{(0)}\),
\[
G(x)
=
W^{(0)}\Phi(x)
=
W^{(0)}\Phi(x')
=
G(x').
\]
Part (i) then implies
\[
f(x)=f(x')
\]
for every downstream map \(T\). 
This proves part (ii).

Finally, consider the Gaussian first-layer feature map. 
Each coordinate-wise Gaussian basis function has the form
\[
\phi_g(t)
=
\exp\!\left(-\frac{(t-c_g)^2}{\epsilon^2}\right),
\qquad
c_g\in[0,1].
\]
For fixed \(t\in[0,1]\) and fixed center \(c_g\), we have
\[
\frac{(t-c_g)^2}{\epsilon^2}\to 0
\qquad
\text{as }
\epsilon\to\infty.
\]
Therefore,
\[
\phi_g(t)\to \exp(0)=1.
\]
Since this holds for every coordinate and every center, the full first-layer feature vector satisfies
\[
\Phi(x)\to \mathbf{1}_{dG}
\qquad
\text{for every }x\in[0,1]^d.
\]
On the finite sample set \(\mathcal{X}=\{x_j\}_{j=1}^N\), stacking the feature vectors row-wise gives
\[
\Phi
\to
\mathbf{1}_{N}\mathbf{1}_{dG}^{\top}.
\]
Consequently,
\[
K_0
=
\Phi\Phi^\top
\to
\left(\mathbf{1}_{N}\mathbf{1}_{dG}^{\top}\right)
\left(\mathbf{1}_{dG}\mathbf{1}_{N}^{\top}\right)
=
\left(\mathbf{1}_{dG}^{\top}\mathbf{1}_{dG}\right)
\mathbf{1}_{N}\mathbf{1}_{N}^{\top}
=
dG\,\mathbf{1}_{N}\mathbf{1}_{N}^{\top}.
\]
This limiting matrix has rank one when \(N>1\). 
Furthermore,
\[
G(x)
=
W^{(0)}\Phi(x)
\to
W^{(0)}\mathbf{1}_{dG},
\]
which is independent of \(x\). 
Thus, the first hidden representation becomes constant across the sample set. 
This proves part (iii).
\end{proof}

\section{Proof of Proposition~\ref{prop:gram_vs_feature_conditioning}}
\label{app:gram_feature_conditioning_proof}

\begin{proof}
Let the singular value decomposition of \(\Phi\in\mathbb{R}^{N\times M}\) be
\[
\Phi
=
U\Sigma V^\top,
\]
where the nonzero diagonal entries of \(\Sigma\) are the singular values
\[
\sigma_1(\Phi)\ge \sigma_2(\Phi)\ge \cdots \ge \sigma_r(\Phi)>0.
\]
Then
\[
\Phi^\top\Phi
=
V\Sigma^\top U^\top U\Sigma V^\top
=
V\Sigma^\top\Sigma V^\top.
\]
Therefore, the eigenvalues of \(\Phi^\top\Phi\) are the squared singular values of \(\Phi\), namely
\[
\lambda_j(\Phi^\top\Phi)
=
\sigma_j(\Phi)^2.
\]

If \(\Phi\) has full column rank, then all \(M\) singular values are positive. Hence,
\[
\kappa(\Phi^\top\Phi)
=
\frac{\lambda_{\max}(\Phi^\top\Phi)}
{\lambda_{\min}(\Phi^\top\Phi)}
=
\frac{\sigma_1(\Phi)^2}
{\sigma_M(\Phi)^2}
=
\left(
\frac{\sigma_1(\Phi)}
{\sigma_M(\Phi)}
\right)^2
=
\kappa(\Phi)^2.
\]
This proves part (i).

If \(\Phi\) is rank deficient, then at least one singular value of \(\Phi\) is zero. 
Therefore, \(\Phi^\top\Phi\) has at least one zero eigenvalue and is singular. 
Thus,
\[
\kappa(\Phi)=\infty,
\qquad
\kappa(\Phi^\top\Phi)=\infty.
\]
This proves part (ii).
\end{proof}

\section{Derivation of the Mat\'ern-5 Lower Scale}
\label{app:matern_lower_bound}

This appendix derives the lower endpoint used in \eqref{eq:matern5_safe_interval_main}. 
The centers are uniformly distributed on \([0,1]\), so the adjacent-center spacing is
\begin{equation}
h=\frac{1}{G-1}.
\label{eq:app_matern_spacing}
\end{equation}

The lower scale is defined by enforcing the same adjacent-center overlap level used in the Gaussian case. 
Specifically, we require
\begin{equation}
\phi(h;\epsilon_{\mathrm{low}})=e^{-1},
\label{eq:app_overlap_rule}
\end{equation}
meaning that the value of a basis function at the neighboring center is \(e^{-1}\) of its peak value.

For the Mat\'ern-5 basis,
\begin{equation}
\phi_{\mathrm{M5}}(d;\epsilon)
=
\left(1+r+\frac{r^2}{3}\right)e^{-r},
\qquad
r=\frac{\sqrt{10}\,d}{\epsilon}.
\label{eq:app_matern5_profile}
\end{equation}
At the neighboring center \(d=h\), define
\begin{equation}
s=\frac{\sqrt{10}\,h}{\epsilon_{\mathrm{low}}^{\mathrm{M5}}}.
\label{eq:app_scaled_distance}
\end{equation}
Substituting \eqref{eq:app_scaled_distance} into \eqref{eq:app_overlap_rule} gives
\begin{equation}
\left(1+s+\frac{s^2}{3}\right)e^{-s}=e^{-1}.
\label{eq:app_matern5_root_equation}
\end{equation}
The positive solution of \eqref{eq:app_matern5_root_equation} is
\begin{equation}
s\approx 2.90463.
\label{eq:app_matern5_root_value}
\end{equation}
Solving \eqref{eq:app_scaled_distance} for \(\epsilon_{\mathrm{low}}^{\mathrm{M5}}\) yields
\begin{equation}
\epsilon_{\mathrm{low}}^{\mathrm{M5}}
=
\frac{\sqrt{10}}{s}\,h
=
\frac{\sqrt{10}}{2.90463}\,h
\approx
1.0887\,h.
\label{eq:app_matern5_lower}
\end{equation}
Using \(h=1/(G-1)\), we obtain
\begin{equation}
\epsilon_{\mathrm{low}}^{\mathrm{M5}}
\approx
\frac{1.0887}{G-1}.
\label{eq:app_matern5_lower_final}
\end{equation}
This gives the lower endpoint used for the Mat\'ern-5 safe interval.

\end{document}